\documentclass[draftcls,onecolumn,11pt]{IEEEtran}
\usepackage{amsmath,amstext,amsfonts,amssymb,dsfont,bbm,hyperref}
\usepackage{etex}
\usepackage{amssymb}
\usepackage{amsmath}
\usepackage{amsthm}
\usepackage{subfigure, epsfig}
\usepackage{pstricks}
\usepackage{pstricks-add}
\usepackage{pst-plot}
\usepackage{graphicx}
\usepackage{pseudocode}
\usepackage{bbm}
\usepackage{acronym}
\usepackage{ctable}
\usepackage{pifont}
\usepackage[ruled,boxed,commentsnumbered]{algorithm2e}
\newtheorem{definition}{Definition}
\newtheorem{thm}[definition]{Theorem}
\newtheorem{cor}[definition]{Corollary}
\newtheorem{lem}[definition]{Lemma}
\newtheorem{prop}[definition]{Proposition}

\newtheorem{remark}{Remark}

\acrodef{iid}[i.i.d.]{independent identically distributed}
\acrodef{wrt}[w.r.t.]{with respect to}
\acrodef{SINR}[SINR]{signal-to-interference plus noise ratio}
\acrodef{SNR}[SNR]{signal-to-noise ratio}
\acrodef{DMC}[DMC]{discrete memoryless channel}
\acrodef{BSC}[BSC]{binary symmetric channel}
\acrodef{KKT}[KKT]{Karush-Kuhn-Tucker}
\acrodef{NPC}[NPC]{Nash-equilibrium power control}
\acrodef{SPC}[SPC]{Semi-coordinated power control}
\acrodef{CCPC}[CCPC]{Costless-communication power control}
\acrodef{CPC}[CPC]{Coded power control}

\newcommand{\eqdef}{\triangleq}
\newcommand{\calX}{\mathcal{X}}
\newcommand{\calU}{\mathcal{U}}
\newcommand{\calY}{\mathcal{Y}}
\newcommand{\calT}{\mathcal{T}}

\newcommand{\ul}{\underline}
\newcommand{\ol}{\overline}

\newcommand{\mc}{\mathcal}
\newcommand{\tc}{\textcolor}

\newcommand{\ds}{\displaystyle}

\newcommand{\idest}{{i.e.,} }
\newcommand{\eg}{e.g., }

\graphicspath{{graphics/}}

\title{Coordination in Distributed Networks via Coded Actions with
Application to Power Control}
 \author{
 \IEEEauthorblockN{Benjamin~Larrousse, Samson~Lasaulce, and Matthieu~R.~Bloch\\}
   \IEEEauthorblockA{\{larrousse, lasaulce\}@lss.supelec.fr, matthieu.bloch@ece.gatech.edu}
     \thanks{This paper was presented in part at the 2013 IEEE International Symposium on Information Theory~\cite{Larrousse-isit2013}, the 2013 International Workshop on Stochastic Methods in Game Theory \cite{lasaulce-erice-2013}, and the 12th IEEE International Symposium on Modeling and Optimization in Mobile, Ad Hoc, and Wireless Networks in 2014~\cite{larrousse-wiopt-2014}. The work of B. Larrousse and S. Lasaulce was supported in part by the French agency ANR through the project LIMICOS - ANR-12-BS03-0005. The work of M. Bloch was supported in part by NSF under award CCF-1320304.}}
     \date{\today}

\begin{document}
\maketitle

\begin{abstract}
This paper investigates the problem of coordinating several agents through their actions, focusing on an asymmetric observation structure with two agents. Specifically, one agent knows the past, present, and future realizations of a state that affects a common payoff function, while the other agent either knows the past realizations of nothing about the state. In both cases, the second agent is assumed to have strictly causal observations of the first agent's actions, which enables the two agents to coordinate. These scenarios are applied to distributed power control; the key idea is that a transmitter may embed information about the wireless channel state into its transmit power levels so that an observation of these levels, \eg the signal-to-interference plus noise ratio, allows the other transmitter to coordinate its power levels. The main contributions of this paper are twofold. First, we provide a characterization of the set of feasible average payoffs when the agents repeatedly take long sequences of actions and the realizations of the system state are \acs{iid}. Second, we exploit these results in the context of distributed power control and introduce the concept of coded power control. We carry out an extensive numerical analysis of the benefits of coded power control over alternative power control policies, and highlight a simple yet non-trivial example of a power control code.
\end{abstract}


\section{Introduction}
\label{sec:introduction}

The main technical problem studied in this paper is the following. Given an integer $N\geq 1$, three discrete alphabets $\calX_0$, $\calX_1$, $\calX_2$, and a \emph{stage payoff} function $w:\calX_0\times\calX_1\times\calX_2\rightarrow \mathbb{R}$, one wants to maximize the average payoff
\begin{equation}
\label{eq:FT-intro}
W_N(x_0^N, x_1^N, x_2^N) \eqdef \frac{1}{T} \sum_{n=1}^N w(x_{0,n}, x_{1,n},x_{2,n})
\end{equation}
\ac{wrt} the sequences $x_1^N\eqdef(x_{1,1},\dots,x_{1,N})\in\calX_1^N$ and $x_2^N\eqdef(x_{2,1},\dots,x_{2,N})\in\calX_2^N$ given the knowledge of $x_0^N \eqdef (x_{0,1},\dots,x_{0,N})\in\calX_0^N$. Without further restrictions and with instantaneous knowledge of $x_{0,n}$, solving this optimization problem consists in finding one of the optimal pairs of variables $(x_{1,n}^\star, x_{2,n}^\star)$ for every $n$. The corresponding maximum value\footnote{This ideal situation is referred to as the ``costless communication'' case. In Section \ref{sec:coded-power-control}, the corresponding power control scenario is called costless communication power control (CCPC).} of $W_N$ is then
\begin{equation}
W_N^{\star} =  \frac{1}{N} \sum_{n=1}^N   \max_{x_1, x_2} w(x_{0,n}, x_1, x_2).
\end{equation}
We assume here that the variable $x_2$ cannot be controlled or optimized directly. As formally described in Section~\ref{sec:problem-statement}, the variable $x_2$ results from imperfect observations of $x_0$ through $x_1$, which induces an \emph{information constraint} in the aforementioned optimization problem. One contribution in Section~\ref{sec:information-constraints} is to precisely characterize this constraint for large $N$ when $x_0^N$ consists of \ac{iid} realizations of a given random variable $X_0$.

This setting is a special case of \emph{distributed optimization}, in which $K$ agents\footnote{In other disciplines such as computer science, control, or economics, agents are sometimes called nodes, controllers, or decision-makers.} connected via a given observation structure have the common objective of maximizing the average payoff $W_N$ for large $N$. The variable $x_k$ with $k\in\{1,\dots,K\}$ is the \emph{action} of Agent $k$ and represents the only variable under its control. The variable $x_0$ is outside of the agents' control and typically represents the realization of a random system state. The observation structure defines how the agents interact through observations of the random state and of each other's actions. The average payoff then measures the degree of \emph{coordination} between the agents, under the observation constraints of the actions imposed by the observation structure. As a concrete example, we apply this framework to power control in Section~\ref{sec:coded-power-control}, in which $x_0$ represents the global wireless channel state information and $x_k$ the power level of Transmitter $k$.

A central question is to characterize the possible values of the average payoff $W_N$ when the agents interact many times, \idest when $N$ is large. Answering this question in its full generality still appears out of reach, and the present paper settles for a special case with $K=2$ agents. Specifically, we assume that Agent 1 has perfect knowledge of the past, current, and future realizations of the state sequence $x_0^N$, while Agent 2 obtains imperfect and strictly causal observations of Agent 1's actions and possesses either strictly causal or no knowledge of the realizations of the state. Despite these restricting assumptions, one may extract valuable concepts and insights of practical relevance from the present work, which can be extended to the general case of $K\geq2$ agents and arbitrary observation structures.

\subsection{Related work}
\label{sec:related-work}
In most of the literature on agent coordination, including classical team decision problems~\cite{basar-book-2013}, agents coordinate their actions through \emph{dedicated channels}, which allow them to signal or communicate with each other without affecting the payoff function. The works most closely related to the present one are~\cite{Cuff2010,Cuff2013}, in which the authors introduce the notions of \emph{empirical and strong coordination} to measure agents' ability to coordinate their actions in a network with noiseless dedicated channels. Empirical coordination measures an average coordination over time and requires the joint empirical distribution of the actions to approach a target distribution asymptotically in variational distance; empirical coordination relates to the communication of probability distributions~\cite{KS07} and tools from rate-distortion theory. Strong coordination is more stringent and asks the distribution of \emph{sequences} of actions to be asymptotically indistinguishable from sequences of actions drawn according to a target distribution, again in terms of variational distance; strong coordination relates to the notion of channel resolvability~\cite{Han1993}. The goal is then to establish the \emph{coordination capacity}~\cite{Cuff2010}, which relates the achievable joint distributions of actions to the fixed rate constraints on the noiseless dedicated channels. The results of~\cite{Cuff2010,Cuff2013} have been extended to a variety of networks with dedicated channels~\cite{Bereyhi2013,Haddadpour2012,Bloch2013a,Bloch2014a}, and optimal codes have been designed for specific settings~\cite{Blasco-Serrano2012,Bloch-2012,Chou2015}.

Much less is known about the coordination via the actions of agents in the absence of dedicated channels, which is the main focus of the present work. The most closely related work is~\cite{Gossner-2006}, in which the authors characterize the set of possible average payoffs for two agents, assuming that each agent perfectly monitors the other agent's actions; the authors establish the set of \emph{implementable distributions}, which are the achievable empirical joint distributions of the actions under the assumed observation structure. In particular, this set is characterized by an information constraint that captures the observation structure between the agents. While~\cite{Gossner-2006} largely relies on combinatorial arguments,~\cite{Cuff-2011} provides an information-theoretic approach of coordination via actions under the name of \emph{implicit communication}. \tc{black}{Coordination via actions also relates to earlier works on encoders with cribbing~\cite{Meulen1977,Willems1985,Asnani2011}; in such models, encoders observe the output signals of other encoders, which effectively creates indirect communication channels to coordinate. Another class of relevant models in which agent actions influence communication are channels with action-dependent states~\cite{Weissman2010}, in which the signals emitted by an agent influence the state of a communication channel.}

To the best of our knowledge, the present work is the first to exploit coordination via actions for distributed resource allocation in wireless networks, and specifically here for distributed power control over an interference channel and multiple-access channels. Much of the distributed power control literature studies the performance of power control schemes using game-theoretic tools. One example is the iterative water-filling algorithm~\cite{yu-jsac-2002}, which is an instance of best-response dynamics (BRD), and is applied over a time horizon over which the wireless channel state is constant. One of the main drawbacks of the various implementations of the BRD for power control problems, see \eg~\cite{lasaulce-book-2011,zappone-comlett-2011,bacci-tsp-2013}, is that they tend to converge to \ac{NPC} policies. The latter are typically Pareto-inefficient, meaning that there exist some schemes \tc{black}{that} would allow all the agents to improve their individual utility \ac{wrt} the \ac{NPC} policies. Another drawback is that such iterative schemes do not always converge. Only restrictive sufficient conditions for convergence are known, see \eg~\cite{scutari-tsp-2009} for the case of multiple input multiple output (MIMO) interference channels, and are met with probability zero for some special cases such as the parallel multiple-access channels~\cite{mertikopoulos-jsac-2012}. In contrast, one of the main benefits of \emph{coded power control} developed in Section~\ref{sec:coded-power-control} is precisely to obtain efficient operating points for the network. This is made possible by having the transmitters exchange information about the quality of the communication links through observed quantities, such as the \ac{SINR}. The \acp{SINR} of the different users effectively act as the outputs of a channel over which transmitters communicate to coordinate their actions. A transmitter codes several realizations of the wireless channel state into a sequence of power levels, which then allows other transmitters to exploit their corresponding sequence of \acp{SINR} to select their power levels. No iterative procedure is required and convergence issues are therefore avoided. We focus our study on efficiency, and \ac{NPC} is therefore compared to coded power control in terms of average sum-rate; other aspects such as wireless channel state information availability and complexity should also be considered but are deferred to future work.

\subsection{Contributions}
\label{sec:main-contributions}

The contributions of the present work are as follows.
\begin{itemize}
\item The results in Section~\ref{sec:information-constraints} extend~\cite{Gossner-2006} by relaxing assumptions about the observation structure. While~\cite{Gossner-2006} assumes that Agent $2$ perfectly monitors the actions of Agent $1$, we consider the case of imperfect monitoring and analyze situations in which Agent $2$ has a strictly causal knowledge (Theorem~\ref{thm:upper-bound} and Corollary~\ref{thm:achiev-ISIT}) or no knowledge (Theorem~\ref{thm:achiev-GP}) of the state.
\item We clarify the connections between the game-theoretic formulation of~\cite{Gossner-2006} and information-theoretic considerations from the literature on state-dependent channels~\cite{GP-1980,Kim2008,Choudhuri2010,Choudhuri2011,Choudhuri2012a}, separation theorems, and empirical coordination~\cite{Cuff2010,Cuff2011a}. We also formulate the determination of the long-run average payoff as an optimization problem, which we study in detail in Section~\ref{sec:optimization-problem} and exploit for power control in Section~\ref{sec:coded-power-control}. 
\item We establish a bridge between the coordination via actions and power control in wireless networks. We develop a new perspective on resource allocation and control, in which designing a resource allocation with high average common payoff amounts to designing a code. Such a code has to strike a balance between sending information about the upcoming realizations of the state, to obtain high payoff in the future, and achieving a good value of the current payoff. As an illustration, we provide a complete description of a power control code for the multiple-access channel in Section~\ref{sec:coord-scheme}.
\end{itemize}


\section{Problem statement}
\label{sec:problem-statement}

For convenience, we provide a summary of the notation used throughout this paper in Table~\ref{tab:notations}.

\begin{table}[b]
\caption{Summary of notation used throughout the paper.}
\label{tab:notations}

\begin{center}
{\normalsize
\begin{tabular}{|c|l|}
\hline
\textbf{Symbol} & \textbf{Meaning}\\
\hline
$Z$  & A  generic random variable\\
$Z_i^j$ & Sequence of random variables $(Z_i,\dots,Z_j)$, $j \geq i$\\
$Z^n$ or $\ul{Z}$ & $Z_i^j$ when $i=1$ and $j=n$\\
$\mathcal{Z}$ & Alphabet of $Z$\\
$|\mathcal{Z}|$ & Cardinality of $\mathcal{Z}$\\
$\Delta(\mathcal{Z})$ & Unit simplex over $\mathcal{Z}$\\
$z$  & Realization of $Z$\\
$z^n$ or $\ul{z}$& Sequence or vector $(z_1,\dots,z_n)$\\
$\mathbb{E}_\mathrm{P}$ & Expectation operator under the probability $\mathrm{P}$\\
$H(Z)$ & Entropy of $Z$\\
$I(Y;Z)$  & Mutual information between $Y$ and $Z$\\
$Z_1 - Z_2 - Z_3$ & Markov chain $\mathrm{P}(z_1|z_2,z_3) = \mathrm{P}(z_1|z_2)$\\
$\mathbbm{1}_{ \{ . \} }$ & Indicator function\\
$\oplus$ & Modulo$-2$ addition\\
$\mathbb{R}_{+}$ & $[0, +\infty)$\\
$\calT^n_\epsilon(Q)$& $\{z^n\in\mathcal{Z}^n:\Vert T_{z^n}-Q\Vert_1<\epsilon\}$\\
$\delta(\epsilon)$& A function of $\epsilon$ such that $\lim_{\epsilon\rightarrow 0}\delta(\epsilon)=0$\\
$\mathrm{T}_{z^n}$ & Type of the sequence $z^n$\\
\hline
\end{tabular}
}
\end{center}
\end{table}



We now formally introduce the problem of interest. We consider $K=2$ agents that have to select their actions repeatedly over $N\geq1$ stages and wish to coordinate via their actions in the presence of a random state and with an observation structure detailed next. At each stage $n\in[1:N]$, the action of Agent $k \in \{1,2\}$ is $x_{k,n}\in\calX_k$ with $|\calX_k| < \infty$, while the realization of the random state is $x_{0,n}\in \calX_0$ with $|\calX_0|< \infty$. The realizations of the state are \ac{iid} according to a random variable $X_0$ with distribution $\rho_0\in\Delta(\calX_0)$. The random state does not depend on the agents' actions but affects a common payoff function\footnote{The
  function $w$ can be any function such that the asymptotic average payoffs defined in the  paper exist.} $w:\calX_0\times\calX_1\times\calX_2\rightarrow\mathbb{R}$. Coordination is measured in terms of the average payoff $W_N(x_0^N,x_1^N,x_2^N) $ as defined in \eqref{eq:FT-intro}. At every stage $n$, Agent $2$ only has access to \emph{imperfect} observations $y_n\in\mathcal{Y}$ of Agent $1$'s actions with $|\mathcal{Y}|<\infty$, which are the output of channel \tc{black}{without memory} and with transition probability
\begin{equation}
\mathrm{P}(y_n |  x_0^{n}, x_1^{n}, x_2^{n}, y^{n-1}) =
\Gamma(y_n | x_{0,n}, x_{1,n}, x_{2,n}) \label{eq:DMC}
\end{equation}
for some fixed conditional probability $\Gamma$. We consider two observation structures defined by the strategies $(\sigma_n)_{1\leq n \leq N}$ and  $(\tau_n)_{1\leq n \leq N}$ of Agents $1$ and $2$, respectively, which restrict how agents observe the state and each other's actions at all stage $n\in [1:N]$ as follows:
\begin{align}\label{eq:strategies-I}
\text{case I:}\quad&\left\{
\begin{array}{ccccc}
\sigma_n^{\mathrm{I}} & : & \calX_0^N & \rightarrow & \calX_1\\
\tau_n^{\mathrm{I}} & : & \calX_0^{n-1} \times \mathcal{Y}^{n-1} & \rightarrow & \calX_2
\end{array}
 \right.\\\label{eq:strategies-II}
\text{case II:}\quad&\left\{
\begin{array}{ccccc}
\sigma_n^{\mathrm{II}} & : & \calX_0^N & \rightarrow & \calX_1\\
\tau_n^{\mathrm{II}} & : & \mathcal{Y}^{n-1} & \rightarrow & \calX_2
\end{array}
 \right..
\end{align}
Note that the strategies differ from conventional block channel coding, since an agent acts at every stage; they may rather be viewed as joint source-channel codes with online coding and decoding. These strategies are also asymmetric since Agent 1 does not observe Agent 2's actions. Symmetric strategies, in which agents would interact, are much more involved and partial results have been recently developed in~\cite{larrousse-itw-2015}. There exist, however, many scenarios, such as cognitive radio, heterogeneous networks, interference alignment, and master-slave communications~\cite{li-icassp-2014} in which asymmetric strategies are relevant. Our objective is to characterize the set of average payoffs that are asymptotically feasible, \idest the possible values for $\lim_{N \rightarrow \infty} \frac{1}{N}\sum_{n=1}^N w(x_{0,n},x_{1,n},x_{2,n})$ under the observation structures defined through~\eqref{eq:strategies-I} and~\eqref{eq:strategies-II}. The definition of the two corresponding feasible sets is as follows.

\begin{definition}[Feasible sets of payoffs]
The feasible set of payoffs in case I is defined as
\begin{align}\label{eq:feasible-set-case-I}
\Omega^{\mathrm{I}} = \Big\{ \omega \in \mathbb{R} \ : \ \exists \, (\sigma_n^{\mathrm{I}},\tau_n^{\mathrm{I}})_{1\leq n \leq N},\, \omega  = \lim_{N \rightarrow \infty} \frac{1}{N}\sum_{n=1}^N \mathbb{E} \left[ w\left(X_{0,n},\sigma_n^{\mathrm{I}}(X_0^{N}),\tau_n^{\mathrm{I}}(X_0^{n-1}, Y^{n-1})\right)\right] \Big\}.
\end{align}
The feasible set of payoffs in case II is defined as
\begin{align}\label{eq:feasible-set-case-II}
\Omega^{\mathrm{II}} = \Big\{ \omega \in \mathbb{R} \ : \ \exists \, (\sigma_n^{\mathrm{II}}, \tau_n^{\mathrm{II}})_{1\leq n \leq N},\,  \omega  = \lim_{N \rightarrow \infty} \frac{1}{N}\sum_{n=1}^N \mathbb{E} \left[ w\left(X_{0,n},\sigma_n^{\mathrm{II}}(X_0^{N}),\tau_n^{\mathrm{II}}(Y^{n-1})\right) \right] \Big\}.
\end{align}
\end{definition}

The feasible sets of payoffs are directly related to the set of \emph{empirical coordinations} over the alphabet $  \mc{X} \eqdef \mc{X}_0\times \mc{X}_1 \times \mc{X}_2,$ defined as follows.
\begin{definition}[Type~\cite{Cover:2006:EIT:1146355}]
  Let $N\geq1$. For any sequence of realizations $z^N$ of the generic random variable $Z$, the type of $z^N$, denoted by $\mathrm{T}_{z^N}$, is the probability distribution on $\mathcal{Z}$ defined by
  \begin{align}
    \mathrm{T}_{z^N}(z) \stackrel{\triangle}{=} \frac{1}{N}\sum_{n=1}^N\mathbbm{1}_{\{z_{n}=z\}}.
    \label{def:type}
  \end{align}
\end{definition}

\begin{definition}[Empirical coordination~{\cite{Cuff2010}}]
  For $\ell \in \{\mathrm{I}, \mathrm{II}\} $, $\ol{Q}\in\Delta(\mc{X})$ is an achievable empirical coordination if there exists a sequence of strategies $(\sigma_n^\ell, \tau_n^\ell)_{1 \leq n \leq N}$ that generates, together with \tc{black}{$X_0^N$}, a sequence \tc{black}{$X^N \in\mc{X}$} such that
  \begin{align}
    \forall \epsilon>0,\quad \lim_{N\rightarrow\infty} \mathrm{P}(||\mathrm{T}_{X^N}-\ol{Q}||_1>\epsilon)=0,
  \end{align}
  \idest the distance between the histogram of a sequence of actions and $\ol{Q}$ converges in probability to $0$.
\end{definition}
Each feasible set of payoffs is the linear image of the corresponding set of empirical distributions under the expectation operator. A value $\omega$ is asymptotically feasible if there exists an achievable empirical coordination $\ol{Q}$ such that $\omega=\mathbb{E}_{\ol{Q}}[w]  = \sum_{x_0,x_1,x_2} \ol{Q}(x_0,x_1,x_2) w(x_0,x_1,x_2)$. We focus on the characterization of achievable empirical coordinations rather than the direct characterization of the feasible sets \tc{black}{of} payoffs.

\begin{remark}
  The notion of empirical coordination relates to the game-theoretic notion of implementability~\cite{Gossner-2006}. For $\ell \in \{\mathrm{I}, \mathrm{II}\}$, $\ol{Q}\in\Delta(\calX)$ is implementable if there exists a sequence of strategies $(\sigma_n^\ell, \tau_n^\ell)_{1 \leq n \leq N}$, $\ell \in \{\mathrm{I}, \mathrm{II}\}$, that induce at each stage $n$ a joint distribution
\begin{align}
  \mathrm{P}_{X_{0,n},X_{1,n},X_{2,n},Y_n} (x_0,x_1,x_2,y)\eqdef   \Gamma(y|x_0,x_1,x_2)\mathrm{P}_{X_{1,n},X_{2,n}|X_{0,n}}(x_1,x_2|x_0) \rho_0(x_0),\label{eq:form_induced_distribution}
\end{align} 
and that generate, together with the sequence $X_0^N$, the sequence $X^N \in\mc{X}$ such that
  \begin{align}
    \lim_{N\rightarrow\infty} ||\mathbb{E}(\mathrm{T}_{X^N})-\ol{Q}||_1=0,
  \end{align}
  \idest the average histogram of a sequence of actions is arbitrarily close to $\ol{Q}$. As shown in Appendix \ref{sec:proof-prop-empirical-implementable}, if $\ol{Q} \in \Delta(\mc{X})$ is an achievable empirical coordination, then it is implementable.
\end{remark}

We conclude this section by a brief discussion of the model, especially Agent $1$'s strategy in~\eqref{eq:strategies-I} and~\eqref{eq:strategies-II} that exploits non-causal knowledge of an \ac{iid} state. This assumption has been often used since the work of Gel'fand and Pinsker~\cite{GP-1980}, but we provide here additional justifications motivated by the application to power control in Section~\ref{sec:coded-power-control}. First, even if Agent~1 only knows future realizations over a limited time horizon, coordination may be significantly improved compared to conventional approaches, such as implementing single-stage game Nash equilibrium-type distributed policies~\cite{yu-jsac-2002,lasaulce-book-2011,scutari-tsp-2009,goodman-pcom-2000}. For instance, power control is typically based on a training phase and an action phase, assuming that a single channel state is known in advance; this corresponds to $N=2$ in our model and, as illustrated in Fig.~\ref{fig1111}, a simple coordination strategy is for Agent $1$ to inform Agent $2$ about the upcoming channels state during odd stages\footnote{For example, Transmitter $1$ might use a high (resp. low) power level on an odd stages to inform Transmitter $2$ that the channel is good (resp. bad) in the next even stage.} and coordinate their actions during even ones. In that context, assuming that Agent $1$ knows the state non-causally is a way to establish an upper bound on the performance all strategies with limited time horizon. Second, predicting the wireless channel state over a long time horizon has recently become realistic. For instance, the trajectory of a mobile user can be forecast~\cite{fourestie-patent-2007,olama-jasp-2006,malmirchegini-twc-2012}, which makes our approach relevant when the wireless channel state is interpreted as the variation of path loss and shadowing. References~\cite{fourestie-patent-2007,olama-jasp-2006,malmirchegini-twc-2012} also suggest that, by sampling the channel at the appropriate rate, the state is nearly \ac{iid}. Finally, note that the proposed approach also applies if the state is only \ac{iid} from block to block, where a block consists of several stages, and suggests that gains can be obtained by varying the power level from stage to stage, even if the channel is constant over a block. 

\begin{figure}[!ht]
  \begin{center}
    \includegraphics{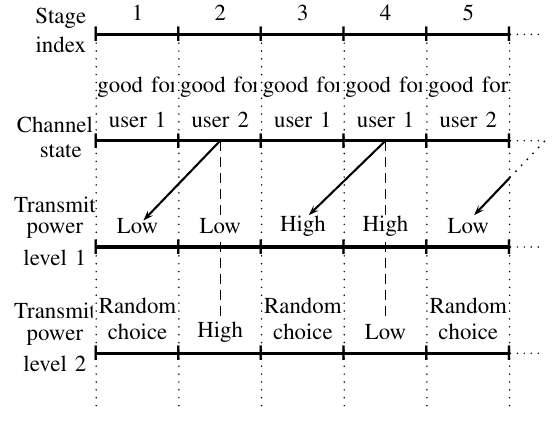}
    \caption{This figure illustrates a simple coordination scheme between two transmitters (which are the agents) in a simplified scenario inspired by~{\cite{Gossner-2006}} where the alphabets are binary: $\mc{X}_0 = \{\text{good for user }1, \text{good for user }2 \} $, $\mc{X}_1 = \{\text{low}, \text{high} \}$, $\mc{X}_2 = \{\text{low}, \text{high} \}$. The informed transmitter (\idest $1$) chooses the lowest (resp. highest) transmit power on the current stage $2t'+1$ if the upcoming wireless channel state on stage $2t'+2$ is good for user $2$ (resp. $1$). If Transmitter $2$ can perfectly retrieve the power levels of Transmitter $1$, it therefore knows the realization of the wireless channel state on stages whose index is even. It transmits at low (resp. high) power if the channel is good for user $1$ (resp. $2$). For stages whose index is odd, it chooses its power at random.}
\label{fig1111}
\end{center}
\end{figure}

\section{Information constraints on \tc{black}{achievable empirical coordination}}
\label{sec:information-constraints}

We first characterize the sets of \tc{black}{achievable empirical coordinations} $\ol{Q}\in \Delta(\calX)$ for the strategies \eqref{eq:strategies-I} and \eqref{eq:strategies-II}.  We show that these sets consist of distributions in $\Delta(\calX)$ subject to an \emph{information constraint} that captures the restrictions imposed by the observation structure. We provide a necessary condition for achievability in Theorem~\ref{thm:upper-bound} and sufficient conditions for strategies \eqref{eq:strategies-I} and \eqref{eq:strategies-II} in Theorem~\ref{thm:achiev-GP} and Corollary~\ref{thm:achiev-ISIT}, respectively.

\subsection{A necessary condition for \tc{black}{achievability}}
\label{sec:necess-cond-impl}

\begin{thm}\label{thm:upper-bound} 
 In both cases I and II, if $\ol{Q}\in\Delta(\calX)$ is an \tc{black}{achievable empirical coordination} then it must be the marginal of $Q\in\Delta(\calX\times\calY)$ factorizing as
\begin{align}
Q(x_0,x_1,x_2,y) = \Gamma(y|x_0,x_1,x_2) \ol{Q}(x_1,x_2|x_0,)\rho_0(x_0), \label{eq:factorization_Q} 
\end{align}
and satisfying the information constraint
\begin{equation}
\label{eq:upper-bound}
 I_{Q} (X_0 ; X_2) \leq I_{Q} (X_1;Y |X_0,X_2).
\end{equation}
\end{thm}

\begin{IEEEproof}
  Since the strategies of case II are special cases of strategies for case I, we derive the necessary conditions by considering strategies for case I, in which Agent 2 has causal knowledge of the state $X_0$. \tc{black}{Let $\ol{Q}\in \Delta(\calX)$ be \tc{black}{an achievable empirical coordination}. Note that 
\begin{align*}
      \mathbb{E}\left(\mathrm{T}_{X_0^N X_1^N X_2^N}(x_0,x_1,x_2)\right)&=\frac{1}{N}\sum_{n=1}^N\mathbb{E}\left(\mathbbm{1}_{\{X_{0,n}, X_{1,n}, X_{2,n}=(x_0,x_1,x_2)\}}\right)=\frac{1}{N}\sum_{n=1}^N \mathrm{P}_{X_{0,n}, X_{1,n}, X_{2,n}}(x_0,x_1,x_2),
\end{align*}
where $\mathrm{P}_{X_{0,n},X_{1,n},X_{2,n},Y_n}$ is defined in \eqref{eq:form_induced_distribution}. It follows from Appendix~\ref{sec:proof-prop-empirical-implementable} that for $\ell \in \{\mathrm{I}, \mathrm{II}\}$, there exists a pair $(\sigma_n^{\ell}, \tau_n^{\ell})_{1\leq n \leq N}$ such that for all $(x_0,x_1,x_2) \in \mathcal{X}$
\begin{align}
  \lim_{N\rightarrow\infty} \frac{1}{N} \sum_{n=1}^{N}\sum_{y\in\calY} \mathrm{P}_{X_{0,n}, X_{1,n}, X_{2,n}, Y_n}(x_0,x_1,x_2,y)  = \ol{Q}(x_0,x_1,x_2).
\end{align}
Because of the specific form of \eqref{eq:form_induced_distribution}, this also implies that
\begin{align}
\lim_{N\rightarrow\infty} \frac{1}{N} \sum_{n=1}^{N} \mathrm{P}_{X_{0,n}, X_{1,n}, X_{2,n}, Y_n}(x_0,x_1,x_2,y) = {Q}(x_0,x_1,x_2,y)\label{eq:limit_assumption}
\end{align}
with $Q$ as in \eqref{eq:factorization_Q}. }
\tc{black}{{The core of the proof consists in establishing an information constraint on the generic joint distribution $\frac{1}{N}\ds{\sum_{n=1}^N}\mathrm{P}_{X_{0,n},X_{1,n},X_{2,n},Y_n}$. We start by expressing the quantity $H(X_0^N)$ in two different manners.} {On one hand we have that}}
\tc{black}{
\begin{align}
H(X_0^N) &= I(X_0^N ; X_0^N, Y^N)\\
&= \sum_{n=1}^{N} I(X_{0,n}; X_0^N, Y^N | X_{0,n+1}^N)\\
& = \sum_{n=1}^{N} \left(I(X_{0,n}; X_0^{n-1}, Y^{n-1} | X_{0,n+1}^N) + I(X_{0,n}; X_{0,n}^{N}, Y_{n}^{N} | X_{0,n+1}^N, X_0^{n-1}, Y^{n-1})\right)\\
& \stackrel{(a)}{=}  \sum_{n=1}^{N} \left(I(X_{0,n}; X_0^{n-1}, Y^{n-1} , X_{0,n+1}^N)+ I(X_{0,n}; X_{0,n}^{N}, Y_{n}^{N} | X_{0,n+1}^N, X_0^{n-1}, Y^{n-1})\right)
\end{align}
where $(a)$ follows from the fact that the sequence $X_0^N$ is \ac{iid}. On the other hand we have that}
\begin{align}
H(X_0^N) &= I(X_0^N ; X_0^N, Y^N)\\
& =  I(X_{0,n+1}^N ; X_0^N, Y^N) + I(X_{0}^n ; X_0^N, Y^N | X_{0,n+1}^N)\\
& = \sum_{n=1}^{N} \left(I(X_{0,n+1}^N ; X_{0,n}, Y_{n} | X_0^{n-1}, Y^{n-1}) + I(X_{0}^n ; X_{0,n}, Y_{n} | X_{0,n+1}^N, X_0^{n-1}, Y^{n-1})\right).
\end{align}
\tc{black}{Since $0\leq $
  \begin{align*}
I(X_{0,n}; X_{0,n}^{N}, Y_{n}^{N} | X_{0,n+1}^N, X_0^{n-1}, Y^{n-1})  &=I(X_{0}^n; X_{0,n},Y_n, Y_{n+1}^{N} | X_{0,n+1}^N, X_0^{n-1}, Y^{n-1}) \\
    &=  I(X_{0}^n ; X_{0,n}, Y_{n} | X_{0,n+1}^N, X_0^{n-1}, Y^{n-1}),
  \end{align*}
we obtain
\begin{equation}
\sum_{n=1}^{N} I(X_{0,n}; X_0^{n-1}, Y^{n-1}, X_{0,n+1}^N) =\sum_{n=1}^{N}  I(X_{0,n+1}^N; X_{0,n}, Y_n | X_0^{n-1}, Y^{n-1}). \label{eq:constraint}
\end{equation}}
\tc{black}{Introducing the uniform random variable $Z\in\{1,\cdots,N\}$ independent of all others, we rewrite \eqref{eq:constraint} as
  \begin{equation} \label{eq:equality-lb-ub}
I(X_{0,Z};X_{0,Z+1}^N, X_0^{Z-1}, Y^{Z-1}|Z) =  I(X_{0,Z+1}^N; X_{0,Z}, Y_Z | X_0^{Z-1}, Y^{Z-1},Z).
  \end{equation}}
We first lower bound the left hand side of (\ref{eq:equality-lb-ub}). Since $X_{0,Z}$ is independent of $Z$ we have that
\begin{align}
  I(X_{0,Z};X_{0,Z+1}^N, X_0^{Z-1}, Y^{Z-1}|Z) = I(X_{0,Z};X_{0,Z+1}^N, X_0^{Z-1}, Y^{Z-1},Z),
\end{align}
which expands as 
  \begin{multline}
 I(X_{0,Z};X_{0,Z+1}^N, X_0^{Z-1}, Y^{Z-1},Z) =  I(X_{0,Z};X_0^{Z-1}, Y^{Z-1},Z)\\ +   I(X_{0,Z};X_{0,Z+1}^N,   | X_0^{Z-1}, Y^{Z-1},Z).\label{eq:expansion_1}
\end{multline}
\tc{black}{By definition, $X_{2,Z}$ is a function of $(Z,X_{0}^{Z-1},Y^{Z-1})$. Consequently,
\begin{align}
      I(X_{0,Z};X_0^{Z-1}, Y^{Z-1},Z) =     I(X_{0,Z};X_{2,Z}, X_0^{Z-1}, Y^{Z-1},Z) \geq  I(X_{0,Z};X_{2,Z}).
\end{align}}
\tc{black}{This gives us the desired lower bound for the left term of (\ref{eq:equality-lb-ub}). We now upper bound for the right hand side of (\ref{eq:equality-lb-ub}) as follows. Using a chain rule, we have
\begin{multline}\label{eq:last-term}
I(X_{0,Z+1}^N; X_{0,Z}, Y_Z | X_0^{Z-1}, Y^{Z-1},Z) = 
I(X_{0,Z+1}^N; X_{0,Z} | X_0^{Z-1}, Y^{Z-1},Z) + \\
I(X_{0,Z+1}^N;  Y_Z | X_0^{Z-1}, X_{0,Z}, Y^{Z-1},Z). 
\end{multline}}
\tc{black}{The last term of (\ref{eq:last-term}) is upper bounded as 
\begin{multline}
  I(X_{0,Z+1}^N; Y_Z | X_0^{Z-1},X_{0,Z},  Y^{Z-1},Z)\\
  \begin{split}
    &=   H(Y_Z | X_0^{Z-1},X_{0,Z},  Y^{Z-1},Z) - H(Y_Z | X_0^{Z-1},X_{0,Z}, X_{0,Z+1}^N Y^{Z-1},Z) \\
    &\stackrel{(b)}{=}H(Y_Z | X_{2,Z},X_0^{Z-1},X_{0,Z},  Y^{Z-1},Z) - H(Y_Z | X_{1,Z},X_{2,Z},X_0^{Z-1},X_{0,Z}, X_{0,Z+1}^N Y^{Z-1},Z) \\
    &\leq H(Y_Z | X_{0,Z}, X_{2,Z}) - H(Y_Z | X_{0,Z}, X_{1,Z},X_{2,Z})\\
    &= I(X_{1,Z};Y_{Z}| X_{0,Z}, X_{2,Z}). \label{eq:upper_bound_rhs}
  \end{split}
\end{multline}
where $(b)$ holds because $X_{2,Z}$ is a function of $(Z,X_{0}^{Z-1},Y^{Z-1})$ and $X_{1,Z}$ is a function of $(Z,X_{0}^{N})$; the inequality follows because conditioning reduces entropy and the Markov chain $(Z, X_0^{Z-1},X_{0,Z+1}^N, Y^{Z-1})-(X_{0,Z},X_{1,Z},X_{2,Z})-Y_{Z}$ deduced from from (\ref{eq:DMC}). By combining~\eqref{eq:equality-lb-ub}-\eqref{eq:upper_bound_rhs}, we find that
\begin{align*}
   I(X_{0,Z};X_{2,Z})\leq I(X_{1,Z};Y_{Z}| X_{0,Z} X_{2,Z}).
\end{align*}}
\tc{black}{To conclude the proof, note that the joint distribution of $X_{0,Z}$, $X_{1,Z}$, $X_{2,Z}$, and $Y_{Q}$, is exactly the distribution $\frac{1}{N}\ds{\sum_{n=1}^N}
 \mathrm{P}_{X_{0,n},X_{1,n},X_{2,n},Y_n}$ and let us introduce function $\Phi^{\mathrm{I}}$
\begin{equation}
\begin{array}{cccc}
\Phi^{\mathrm{I}}: & \Delta(\calX\times\calY)&\rightarrow& 
\mathbb{R}\\
 & Q &\mapsto &I_{Q} (X_0 ; X_2) - I_{Q} (X_1;Y |X_0,X_2)
\end{array},
\label{eq:def-Phi-I}
\end{equation}}
which is continuous. \tc{black}{Because of~\eqref{eq:limit_assumption}, $\forall \varepsilon'>0$ there exists $N'$ such that $\forall N \geq N'$,
\begin{align}
\Phi^{\mathrm{I}}(Q) \leq \Phi^{\mathrm{I}}\left(\frac{1}{N}\sum_{n=1}^N \mathrm{P}_{X_{0,n},X_{1,n},X_{2,n},Y_n}\right)+\varepsilon'.\label{eq:continuity_Q}
\end{align}}
\end{IEEEproof}
Theorem \ref{thm:upper-bound} has the following interpretation. Agent $2$'s actions are represented by $X_2$ and correspond to a joint source-channel decoding operation with distortion on the information source represented by $X_0$. To be achievable, the distortion rate must not exceed the transmission rate allowed by the channel, whose input and output are represented by Agent $1$'s action $X_1$ and the signal $Y$ observed by Agent $2$. Therefore, the pair $S=(X_0,X_2)$ plays the same role as the side information in state-dependent channels~\cite{NetworkInformationTheory}. Although we exploit this interpretation when establishing sufficient conditions for \tc{black}{achievability} in Section~\ref{sec:implementability}, the argument seems inappropriate to show that the sufficient conditions are also necessary. In contrast to classical arguments in converse proofs for state-dependent channels~\cite{GP-1980,Merhav2003}, in which the transmitted ``message'' is independent of the channel state, here the role of the message is played by the quantity $X_0^N$, which is not independent of $S^N=(X_0^N,X_2^N)$.

%
\subsection{Sufficient conditions for \tc{black}{achievability}}
\label{sec:implementability}
We start by addressing the special case of distributions $\ol{Q} \in\Delta(\calX)$ with marginal $\rho_0\in\Delta(\calX_0)$, for which the distribution $Q(x_0,x_1,x_2,y) = \Gamma(y|x_0,x_1,x_2)  \ol{Q}(x_0,x_1,x_2) $ satisfies $I_Q(X_1;Y|X_0X_2)=0$. By Theorem~\ref{thm:upper-bound}, such $\ol{Q}$ is an achievable empirical coordination only if $I_Q(X_0;X_2)=0$, so that $\ol{Q}$ factorizes as $\ol{Q}(x_1|x_0,x_2)\rho_0(x_0)\ol{Q}(x_2)$. This distribution is trivially achievable by time-sharing between strategies in which: i) Agent 2 plays a fixed action $x_2$; ii) Agent 1 generates actions according to $\ol{Q}(x_1|x_0,x_2)$; iii) playing each strategy with fixed $x_2$ a fraction $\ol{Q}(x_2)$ of the time. Hence, we now focus on distributions $\ol{Q}$ for which $I_Q(X_1;Y|X_0X_2)>0$

We now characterize \tc{black}{achievable empirical coordination} for the observation structure of case II in~\eqref{eq:strategies-II}. 

\begin{thm} \label{thm:achiev-GP} Consider the observation structure in case $\mathrm{II}$. Let $U$ be a random variable whose realizations lie in the alphabet $\mathcal{U}$, $|\mathcal{U}| < \infty$. Let $\ol{Q} \in\Delta(\calX)$ be with marginal $\rho_0\in\Delta(\calX_0)$. If $Q\in\Delta(\calX\times\calY\times\calU)$ defined as
\begin{align}
  Q(x_0,x_1,x_2,y,u) =\mathrm{P}(u|x_0,x_1,x_2) \Gamma(y|x_0,x_1,x_2) \ol{Q}(x_0,x_1,x_2)\label{eq:factorization_case_II}
\end{align}
verifies the constraint
\begin{equation}
\label{eq:lower-bound-II}
 I_{Q} (X_0 ; X_2) < I_{Q} (U;Y,X_2) - I_Q(U;X_0,X_2),
\end{equation}
then $\ol{Q}$ is \tc{black}{an achievable empirical coordination.}.
\end{thm}

\begin{IEEEproof}
Consider a distribution $Q_{X_0X_1X_2YU}\in\Delta(\calX\times\calY\times \mc{U})$ that satisfies~\eqref{eq:factorization_case_II} \tc{black}{and~\eqref{eq:lower-bound-II}}. We denote by $Q_{U X_0 X_2}$, $Q_{U}$, $\overline{Q}_{X_0X_1X_2}$, $Q_{X_0X_2}$, and $Q_{X_2}$ the resulting marginal distributions. 

The crux of the proof is to design strategies from a block-Markov coding scheme that operates over $B$ blocks of $m\geq 1$ actions each. As illustrated in Table~\ref{tab:tablethGP}, in every block $b\in[1:B-1]$, Agent 1 communicates to Agent 2 the actions that Agent 2 should play in block $b+1$. This is possible by restricting the actions played by Agent 2 in each block $b$ to a codebook of actions $\{\ul{x}_2(i_b):i_b\in[1:2^{mR}]\}$, so that Agent 1 \tc{black}{only has} to communicate the index $i_{b}$ to be played in the next block. The problem then essentially reduces to a joint source-channel coding problem over a state-dependent channel, for which in every block $b$:
\begin{itemize}
\item the state is known non-causally by Agent 1, as per the observation structure in~\eqref{eq:strategies-II};
\item Agent 1 communicates with Agent 2 over a state-dependent \tc{black}{discrete channel without memory} and with transition probability $\Gamma(y|x_0,x_1,x_2)$;
\item the channel state consists of state sequence $\ul{x}_0^{(b)}$ and action sequence $\ul{x}_2(\widehat{i}_{b-1})$, where $\widehat{i}_{b-1}$ is the index decoded by Agent 2 at the end of block $b-1$. Agent 1 only knows $\ul{x}_2(i_{b-1})$ but the effect of using $\widehat{i}_{b-1}$ in place of $i_{b-1}$ is later proved to be asymptotically negligible; 
\item the action sequence communicated is $\ul{x}_2(i_{b})$, chosen to be empirically coordinated with the state sequence $\ul{x}_0^{(b+1)}$ of block $b+1$;
\item $i_b$ is encoded through Gel'fand-Pinsker coding into an action sequence $\ul{x}_1^{(b)}$, chosen to be empirically coordinated with $(\ul{x}_0^{(b)},\ul{x}_2(i_{b-1}))$.
\end{itemize}
Intuitively, $R$ must be sufficiently large so that one may find a codeword $\ul{x}_2(i_b)$ coordinated with any state sequence $\ul{x}_0^{(b+1)}$; simultaneously, $R$ must be small enough to ensure that the index $i_b$ is reliably decoded by Agent 2 after transmission over the channel $\Gamma(y|x_0,x_1,x_2)$. The formal analysis of these conditions, which we develop next, establishes the result. 
\smallskip

\begin{table}[h]
  \centering
\caption{Encoding and decoding used in the proof of Theorem~\ref{thm:achiev-GP}}
\label{tab:tablethGP}
  \begin{tabular}[]{c|cccccc}
    Block                     &1&2&$\cdots$&$b$&$\cdots$&$B$\\\hline
    Message                &$i_1$&$i_2$&$\cdots$&$i_b$&$\cdots$&$i_B$\\
    State                     &$\ul{x}_0^{(1)}$&$\ul{x}_0^{(2)}$&$\cdots$&$\ul{x}_0^{(b)}$&$\cdots$&$\ul{x}_0^{(B)}$\\
    Agent 2 action      &$\ul{x}_2^*$&$\ul{x}_2(\widehat{i}_{1})$&$\cdots$&$\ul{x}_2(\widehat{i}_{b-1})$&$\cdots$&$\ul{x}_2(\widehat{i}_{B-1})$\\
    Agent 1 codeword      &$\ul{u}_1(i_{1},j_1)$&$\ul{u}(i_{2},j_2)$&$\cdots$&$\ul{u}(i_{b},j_b)$&$\cdots$&$\ul{u}(i_{B}j_B)$\\ 
    Agent 1 action      &$\ul{x}_1^{(1)}$&$\ul{x}_1^{(2)}$&$\cdots$&$\ul{x}_1^{(b)}$&$\cdots$&$\ul{x}_1^{(B)}$\\ 
    Agent 2 decoding &$\widehat{i}_1$&$\widehat{i}_2$&$\cdots$&$\widehat{i}_b$&$\cdots$&$\widehat{i}_B$\\\hline
  \end{tabular}
\end{table}

Unlike the block-Markov schemes used, for instance, in relay channels, in which all nodes may agree on a \emph{fixed message} in the first block at the expense of a small rate loss, the first block must be dealt with more carefully. In fact, we may have to account for an ``uncoordinated'' transmission in the first block, in which Agent 1 may not know the actions $\ul{x}_2$ of Agent 2 and is forced to communicate at rate $\widehat{R}$ that differs from the rate $R$ used in subsequent blocks. To characterize $\widehat{R}$, we introduce another joint distribution $\widehat{Q}$ that factorizes as
\begin{align}
  \label{eq:factorization_q_hat_II}
  \widehat{Q}(x_0,x_1,x_2,y,u) = \Gamma(y|x_0,x_1,x_2) \mathrm{P}(u|x_0x_1x_2)\ol{Q}(x_1|x_0x_2)\rho_0(x_0) \mathbbm{1}_{\{x_2=x_2^*\}}
\end{align}
and differs from $\ol{Q}$ in that $X_0$ is independent of $X_2$, which is a constant. Assume that $I_{\widehat{Q}} (U;Y,X_2) - I_{\widehat{Q}}(U;X_0,X_2)=0$ for all $P$ and $x_2^*$. In particular, for $U=X_0$, we obtain $I_{\widehat{Q}} (X_0;Y,X_2) - I_{\widehat{Q}}(X_0;X_0,X_2)=0$; this is equivalent to $H_{\widehat{Q}}(X_0|YX_2)=0$, so that $x_0$ must be a function of $y$ and $x_2^*$. For $U=X_1$, we also obtain $I_{\widehat{Q}} (X_1;Y,X_2) - I_{\widehat{Q}}(X_1;X_0,X_2)=0$, which using the previously established fact leads to $I_{\widehat{Q}} (X_1;Y|X_0X_2)=0$. Then, for all $(x_0,x_2)\in\calX_0\times\calX_2$, it must be that $I_{\widehat{Q}}(X_1;Y|X_0=x_0,X_2=x_2)=0$ and therefore  $I_{{Q}}(X_1;Y|X_0=x_0,X_2=x_2)=0$. Consequently, $I_Q(X_1;Y|X_0X_2)=0$, which we have excluded from the analysis. 
Hence, we can assume that there exist $P$ and $x_2^*$ such that $I_{\widehat{Q}} (U;Y,X_2) - I_{\widehat{Q}}(U;X_0,X_2)>0$.

Now, let $\epsilon>0$. Let $R>0$, $R'>0$, $\widehat{R}>0$, $\widehat{R}'>0$, $\frac{\epsilon}{2}>\epsilon_3>\epsilon_2> \epsilon_1>0$ be real numbers and $m \geq 1$ to be specified later. Define 
\begin{align}
  \alpha&\eqdef \max\left(   \left\lceil   \frac{R}{\widehat{R}} \right\rceil,1\right)\label{eq:choice_alpha_case1}\\
  B&\eqdef \left\lceil 1+\alpha \left(\frac{4}{\epsilon}-1\right)\right\rceil.\label{eq:choice_B_case1}
\end{align}
Intuitively, $\alpha$ measures the rate penalty suffered from the uncoordinated transmission at rate $\widehat{R}$ in the first block. The choice of $B$ merely ensures that $\frac{2\alpha}{B-1+\alpha}\leq\frac{\epsilon}{2}$, as exploited later.
\smallskip

\emph{Source codebook generation for $b=1$.} Choose $x_2^*$ such that $I_{\widehat{Q}} (U;Y,X_2) - I_{\widehat{Q}}(U;X_0,X_2)>0$. The actions of Agent 2 are $\ul{x}_2^*$ consisting of $m$ repetitions of $x_2^*$ and revealed to both agents.\smallskip

\emph{Source codebooks generation for $b\in[2:B+1]$.} Randomly and independently generate $2^{m R}$ sequences according to $\ds{\Pi_{n=1}^m} Q_{X_2}(x_{2,n})$, label them $\ul{x}_2(i_b)$ with $i_b\in[1:2^{mR}]$ and reveal them to both agents.\smallskip

\emph{Channel codebook generation for $b=1$.} Randomly and independently generate $2^{\alpha m(\widehat{R}'+\widehat{R})}$ sequences according to $\ds{\Pi_{n=1}^m} \widehat{Q}_{U}(u_{n})$, label them $\ul{u}(i_1,j_1)$ with $i_1\in[1:2^{\alpha m\widehat{R}}]$ and $j_1\in[1:2^{\alpha m\widehat{R}'}]$, and reveal them to both agents.\smallskip

\emph{Channel codebook generation for $b\in[2:B]$.} Randomly and independently generate $2^{m(R'+R)}$ sequences according to $\ds{\Pi_{n=1}^m} Q_{U}(u_{n})$, label them $\ul{u}(i_b,j_b)$ with $i_b\in[1:2^{mR}]$ and $j_b\in[1:2^{mR'}]$, and reveal them to both agents.\smallskip

\emph{Source encoding at Agent 1 in block $b\in[1:B]$.} At the beginning of block $b$, Agent 1 uses its non-causal knowledge of the state $\ul{x}_0^{(b+1)}$ in the next block $b+1$ to look for an index $i_b$ such that $({\ul{x}_0^{(b+1)}, \ul{x}_2(i_b)}) \in\mc{T}_{\epsilon_1}^m(Q_{X_0X_2})$. If there is more than one such index, it chooses the smallest among them, otherwise it chooses $i_b=1$.\smallskip

\emph{Channel encoding at Agent 1 in block $b=1$.} Agent 1 uses its knowledge of $(\ul{x}_0^{(1)},\ul{x}_2^*)$ to look for an index $j_1$ such that 
\begin{align} 
  \left({\ul{u}(i_1,j_1),\ul{x}_0^{(1)}, \ul{x}_2^*}\right) \in\mc{T}_{\epsilon_2}^m(\widehat{Q}_{UX_0X_2})
\end{align}
If there is more than one such index, it chooses the smallest among them, otherwise it chooses $j_1=1$. Finally, Agent 1 generates a sequence $\ul{x}_1^{(1)}$ by passing the sequences $\ul{u}(i_1,j_1)$, $\ul{x}_0^{(1)}$, and $\ul{x}_2^*$ through a \tc{black}{channel without memory and} with transition probability $\widehat{Q}_{X_1|UX_0X_2}$, and transmits it.\smallskip

\emph{Channel encoding at Agent in block $b\in[2:B]$.} Agent 1 uses its knowledge of $(\ul{x}_0^{(b)},\ul{x}_2(i_{b-1}))$ to look for an index $j_b$ such that 
\begin{align} 
  \left({\ul{u}(i_b,j_b),\ul{x}_0^{(b)}, \ul{x}_2(i_{b-1})}\right) \in\mc{T}_{\epsilon_2}^m(Q_{UX_0X_2})
\end{align}
If there is more than one such index, it chooses the smallest among them, otherwise it chooses $j_b=1$. Finally, Agent 1 generates a sequence $\ul{x}_1^{(b)}$ by passing the sequences $\ul{u}(i_b,j_b)$, $\ul{x}_0^{(b)}$, and $\ul{x}_2(i_{b-1})$ through a \tc{black}{channel without memory and} with transition probability $Q_{X_1|UX_0X_2}$, and transmits it.\smallskip

\emph{Decoding at Agent 2 in block $b=1$.} At the end of block $1$, Agent $2$ observes the sequence of channel outputs $\ul{y}^{(1)}$ and knows its sequence of actions $\ul{x}_2^*$  in block $1$. Agent 2 then looks for a pair of indices $(\widehat{i}_{1},\widehat{j}_1)$ such that
 \begin{align}
\left({\ul{u}(\widehat{i}_1,\widehat{j}_1), \ul{y}^{(1)}, \ul{x}_2^*}\right)\in\mc{T}_{\epsilon_3}^m(\widehat{Q}_{UYX_2}). 
\end{align}
If there is none or more than one such index, Agent 2 sets $\widehat{i}_1=\widehat{j}_1=1$.\smallskip

\emph{Channel decoding at Agent 2 in block $b\in[2:B]$.} At the end of block $b$, Agent $2$ observes the sequence of channel outputs $\ul{y}^{(b)}$ and knows its sequence of actions $\ul{x}_2(\widehat{i}_{b-1})$ in block $b$. Agent 2 then looks for a pair of indices $(\widehat{i}_{b},\widehat{j}_b)$ such that
 \begin{align}
\left({\ul{u}(\widehat{i}_b,\widehat{j}_b), \ul{y}^{(b)}, \ul{x}_2(\widehat{i}_{b-1})}\right)\in\mc{T}_{\epsilon_3}^m(Q_{UYX_2}). 
\end{align}
If there is none or more than one such index, Agent 2 sets $\widehat{i}_b=\widehat{j}_b=1$.\smallskip

\emph{Source decoding at Agent 2 in block $b\in[1:B]$}. Agent 2 transmits $\ul{x}_2(\widehat{i}_{b-1})$, where $\widehat{i}_{b-1}$ is its estimate of the message transmitted by Agent 1 in the previous block $b-1$, with the convention that $\ul{x}_2(\widehat{i}_0)=\ul{x}_2^*$.\smallskip

\emph{Analysis.} We prove that $\overline{Q}$ is an achievable empirical coordination. We therefore introduce the event
\begin{align}
  E\triangleq \{(\ul{X}_0^N,\ul{X}_1^N,\ul{X}_2^N)\notin \calT_\epsilon^N(\overline Q)\}
\end{align}
with $N=mB$ and we proceed to show that $\mathrm{P}(E)$ can be made arbitrarily small for $n$ and $B$ sufficiently large and a proper choice of the rates  $R$, $R'$, $\widehat{R}$, and $\widehat{R}'$. We start by introducing the following events..
\begin{align*}
  E_0&\triangleq \{(I_1,J_1)\neq (\widehat{I}_1,\widehat{J}_1)\}\\
\forall b\in[1:N]\quad  E_1^{(b)}&\triangleq \{(\ul{X}_0^{(b+1)},\ul{x}_2(i'_b))\notin \calT^m_{\epsilon_1}(Q_{X_0X_2}) \,\forall \, i'_b\in [1:2^{mR}]\}\\
  E_2^{(b)}&\triangleq \{(\ul{u}^{(b)}(I_b,j'_b),\ul{X}_0^{(b)},\ul{x}_2(I_{b-1}))\notin \calT^m_{\epsilon_2}(Q_{UX_0X_2}) \,\forall \, j'_b \in[1:2^{mR'}]\}\\
  E_3^{(b)}&\triangleq \left\{(\ul{u}(I_b,J_b),\ul{X}_0^{(b)},X_1^{(b)},\ul{x}_2^{(b)}(\widehat{I}_{b-1}),\ul{Y}^{(b)}) \notin \calT^m_{\epsilon_3}(Q_{UX_0X_1X_2Y})\right\}\\
  E_4^{(b)}&\triangleq \left\{(\ul{u}(i'_b,j'_b),\ul{x}_2(\widehat{I}_{b-1}),\ul{Y}^{(b)}) \in \calT^m_{\epsilon_3}(Q)\text{ for some $(i'_b,j'_b)\neq (I_b,J_b)$}\right\}.
\end{align*}
We start by developing an upper bound for $\Vert \mathrm{T}_{\ul{x}_0^N\ul{x}_1^N\ul{x}_2^N}-\overline{Q}\Vert_1$, whose proof can be found in Appendix~\ref{app:theo-GP}.
\begin{lem}
  \label{lmc:bound_norm_1} We have that
  \begin{align}
  \Vert \mathrm{T}_{\ul{x}_0^N\ul{x}_1^N\ul{x}_2^N}-\overline{Q}\Vert_1 \leq \frac{2\alpha}{B-1+\alpha}+\frac{1}{B-1} \sum_{b=2}^{B} \Vert \mathrm{T}_{\ul{x}_0^{(b)}\ul{x}_1^{(b)}\ul{x}_2^{(b)}}-\overline{Q}\Vert_1.
\end{align}
\end{lem}
Recalling the choice of $B$ in~\eqref{eq:choice_B_case1}, we therefore have
\begin{align}
  \mathrm{P}(E)&= \mathrm{P}\left(\Vert \mathrm{T}_{\ul{X}_0^N\ul{X}_1^N\ul{X}_2^N}-\overline{Q}\Vert_1\geq \epsilon\right)\\
    &\leq \mathrm{P}\left(\frac{1}{B-1}\sum_{b=2}^{B}\Vert \mathrm{T}_{\ul{X}_0^{(b)}\ul{X}_1^{(b)}\ul{X}_2^{(b)}}-\overline{Q}\Vert_1\geq \frac{\epsilon}{2}\right)\\
    &\leq \mathrm{P}\left(\Vert   \mathrm{T}_{\ul{X}_0^{(b)}\ul{X}_1^{(b)}\ul{X}_2^{(b)}}-\overline{Q}\Vert_1\geq \frac{\epsilon}{2}\text{ for some } b\in[2:B]\right)\displaybreak[0]\\
    &\leq \mathrm{P}\left(E_0\cup E_1^{(1)}\bigcup_{b=2}^{(b)}\left( E_1^{(b)}\cup E_2^{(b)} \cup E_3^{(b)}\cup E_4^{(b)}\right)\right) \displaybreak[0]\\
  &\leq \mathrm{P}(E_0)+\sum_{b=1}^{B}\mathrm{P}(E_1^{(b)}) +\sum_{b=2}^{B}\mathrm{P}(E_2^{(b)}|E_1^{(b-1)c})\nonumber\\
  &\qquad+\sum_{b=2}^{B}\mathrm{P}(E_3^{(b)}\cap E_1^{(b-1)c}\cap E_2^{(b-1)c}\cap E_3^{(b-1)c} \cap E_4^{(b-1)c} \cap E_0^c)\nonumber\\
  &\qquad +\sum_{b=2}^{B}\mathrm{P}(E_4^{(b)}\cap E_2^{(b-1)c} \cap E_3^{(b-1)c} \cap E_4^{(b-1)c} \cap E_0^c)
\end{align}
As proved in Appendix~\ref{app:theo-GP}, the following lemmas show that all the averages over the random codebooks of the terms above vanish as $n\rightarrow\infty$.

\begin{lem}
  \label{lm:lemma_case_2_event_0}
  If $\widehat{R}>I_{\widehat{Q}}(U;X_0X_2)+\delta(\epsilon_2)$ and $\widehat{R}+\widehat{R}'< I_{\widehat{Q}}(U;YX_2)-\delta(\epsilon_3)$, then
  \begin{align}
    \lim_{n\rightarrow\infty}E\left(\mathrm{P}(E_0)\right)=0.
  \end{align}
\end{lem}
\begin{lem}
  \label{lm:lemma_case_2_event_1}
  If ${R}>I_Q(X_0;X_2)+\delta(\epsilon_1)$, then for any $b\in[1:B]$
  \begin{align}
    \lim_{n\rightarrow\infty}E\left(\mathrm{P}(E_1^{(b)})\right)=0.
  \end{align}
\end{lem}
\begin{lem}
  \label{lm:lemma_case_2_event_2}
  If ${R'}>I_Q(U;X_0,X_2)+\delta(\epsilon_2)$, then for any $b\in[2:B]$
  \begin{align}
    \lim_{n\rightarrow\infty}E\left(\mathrm{P}(E_2^{(b))}|E_1^{(b-1)})\right)=0.
  \end{align}
\end{lem}

\begin{lem}
  \label{lm:lemma_case_2_event_3}
  For any $b\in[2:B]$
  \begin{align}
\lim_{n\rightarrow\infty}E\left(\mathrm{P}(E_3^{(b)}|E_2^{(b)c}\cap E_2^{(b-1)c} \cap E_3^{(b-1)c} \cap E_4^{(b-1)c} \cap E_0^c)\right)=0.
  \end{align}
\end{lem}

\begin{lem}
  \label{lm:lemma_case_2_event_4}
  If ${R+R'}<I_Q(U;Y,X_2)-\delta(\epsilon_3)$, then for any $b\in[2:B]$
  \begin{align}
    \lim_{n\rightarrow\infty}E\left(\mathrm{P}(E_4^{(b)}\cap E_2^{(b-1)c} \cap E_3^{(b-1)c} \cap E_4^{(b-1)c} \cap E_0^c)\right)=0.
  \end{align}
\end{lem}

Hence, we can find $\epsilon_1$, $\epsilon_2$, and $\epsilon_3$ small enough such that $\lim_{n\rightarrow\infty}E(\mathrm{P}(E))=0$. In particular, there must exists at least one sequence of codes such that $\lim_{n\rightarrow\infty}\mathrm{P}(E)=0$. Since $\epsilon>0$ can be chosen arbitrarily small, $\ol{Q}$ is an achievable empirical coordination.
\end{IEEEproof}

A few comments are in order regarding the result in Theorem~\ref{thm:achiev-GP}. The condition $I_Q(X_0;X_2) < I_Q(U;Y,X_2)-I_Q(U;X_0X_2)$ ensures $\widehat{i}_{b-1} = i_{b-1}$ with high probability as $m\rightarrow\infty$, so that the ``side information'' $\ul{x}_2(i_{b-1})$ used by Agent $1$ to correlate its actions is identical to the true actions $\ul{x}_2(\widehat{i}_{b-1})$ of Agent  $2$; hence, Agent $1$ effectively knows the actions of Agent $2$ without directly observing them. Furthermore, the state sequence $\ul{x}_0^{(b+1)}$, which plays the role of the message in block $b$, is independent of the ``side information'' $(\ul{x}_0^{(b)}, \ul{x}_2(i_{b-1}))$; this allows us to reuse classical coding schemes for the transmission of messages over state-dependent channels. However, the proof of Theorem~\ref{thm:achiev-GP} exhibits a key difference with the usual Gel'fand-Pinsker coding scheme~\cite{GP-1980} and its extensions~\cite{Merhav2003}. While using the channel decoder's past outputs does not  improve the channel capacity, it helps for coordination. Specifically, a classical Gel'fand-Pinsker coding results would lead to an information constraint 
\begin{equation}\label{eq:GP-pure}
I_Q(X_0;X_2) < I_Q(U;Y) - I_Q(U;X_0,X_2)
\end{equation}
which is more restrictive than \eqref{eq:lower-bound-II}.

\begin{cor}
  \label{thm:achiev-ISIT} 
Consider the observation structure in case I. Let $\ol{Q} \in\Delta(\calX)$ be with marginal $\rho_0\in\Delta(\calX_0)$. 
If $Q\in\Delta(\calX\times\calY)$ defined as
\begin{align}
  Q(x_0,x_1,x_2,y) = \Gamma(y|x_0,x_1,x_2)  \ol{Q}(x_0,x_1,x_2) \label{eq:factorization_case_1}
\end{align}
satisfies the constraint
\begin{equation}
\label{eq:lower-bound-I}
 I_{Q} (X_0 ; X_2) < I_{Q} (X_1;Y |X_0,X_2),
\end{equation}
then $\ol{Q}$ is \tc{black}{an achievable empirical coordination.}
\end{cor}
\begin{IEEEproof}
  Case I differs from Case II by having the state available strictly causally at Agent 2; we can therefore apply the results of Theorem~\ref{thm:achiev-GP} by providing $X_0$ as a second output to Agent 2. Applying Theorem~\ref{thm:achiev-GP} with $(Y,X_0)$ in place of $Y$, we find that if $Q$ defined as in~\eqref{eq:factorization_case_II} satisfies $ I_{Q} (X_0 ; X_2) < I_{Q} (U;YX_0X_2)-I_{Q} (U;X_0X_2)$, then $\ol{Q}$ is an achievable empirical coordination. Since, $I_{Q} (U;YX_0X_2)-I_{Q} (U;X_0X_2)=I_{Q}(U;Y|X_0X_2)$, setting $U=X_1$ yields the desired result.
\end{IEEEproof}

Setting aside the already discussed case of equality in~\eqref{eq:upper-bound}, the information constraints of Theorem~\ref{thm:upper-bound} and Corollary~\ref{thm:achiev-ISIT} coincide, hence establishing a necessary and sufficient condition for a joint distribution $\ol{Q}\in\Delta(\calX)$ to be implementable in case I and a complete characterization of the associated set of achievable payoffs. This also shows that having Agent 1 select the actions played by Agent 2 and separating source and channel encoding operations do not incur any loss of optimality. We apply this result in Section~\ref{sec:coded-power-control} to an interference network with two transmitters and two receivers, in which Transmitter $1$ may represent the most informed agent, such as a primary transmitter~\cite{li-icassp-2014, haykin-jsac-2005}, $\Gamma$ may represent an \ac{SINR} feedback channel from Receiver $2$ to Transmitter $2$. 

Our results hold under the assumption of perfect monitoring~\cite{Gossner-2006} in which Agent $2$ perfectly monitors the actions of Agent $1$, i.e., $Y = X_1$. Equations \eqref{eq:upper-bound}, \eqref{eq:lower-bound-I}, \eqref{eq:lower-bound-II}, and \eqref{eq:GP-pure} then coincide with the information constraint $I_Q(X_0;X_2) \leq H_Q(X_1|X_0,X_2)$~\cite{Gossner-2006}, confirming as noted in~\cite{Gossner-2006,Cuff-2011} that allowing Agent $1$ to observe the action of the other agent or providing Agent 2 with the past realizations of the state $X_0^{n-1}$ does not improve the set of feasible payoffs under perfect monitoring. However, this observation regarding the set of \emph{feasible payoffs} may not hold for the set of \emph{Nash equilibrium payoffs}, which are relevant when agents have diverging interests and in which case it matters whether an agent observes the actions of the others or not. In a power control setting, if the transmitters implement a cooperation plan that consists in transmitting at low power as long as no transmitter uses a high power level, see \eg \cite{letreust-twc-2010}, it matters if the transmitters are able to check whether the others effectively use a low power level. We focus here on a cooperative setting in which a designer has a precise objective (maximizing the network throughput, minimizing the total network energy consumption, etc.) and wants the terminals to implement a power control algorithm with only local knowledge and reasonable complexity. This setting can be seen as a first step toward analyzing the more general situation in which agents may have diverging interests; this would happen in power control in the presence of several operators.


Note that we have not proved whether the information constraint of Theorem~\ref{thm:achiev-GP} is a necessary condition for implementability in case II. One might be tempted to adopt a side information interpretation of the problem to derive the converse, since~\eqref{eq:lower-bound-II} resembles the situation of~\cite{CoverChiang-2002}; however, finding the appropriate auxiliary variables does not seem straightforward and is left as a refinement of the present analysis.

While Theorem~\ref{thm:achiev-GP} and Corollary~\ref{thm:achiev-ISIT} have been derived for an \ac{iid} random state, the results generalize to a situation in which the state is constant over $L\geq1$ consecutive stages and \ac{iid} from one block of $L$ stages to the next. For strategies as in case I, the information constraint becomes
\begin{align}
  \frac{1}{L} I_Q(X_0;X_2)< I_Q(X_1;Y|X_0,X_2).\label{eq:block_fading}
\end{align}
In fact, one can reproduce the argument in the proof of Corollary~\ref{thm:achiev-ISIT} and remark that one can communicate over the channel at a rate $L$ times larger than the rate required for the covering of the source. Specifically, to encode $m$ realizations of the random state, the source codebooks must contain $ 2^{mR}$ codewords with~$R>I_Q(X_0;X_2)$; however, the channel codebooks can contain $2^{mLR}$ codewords with $R<I_Q(X_1;Y|X_0X_2)$. The source and channel codes are compatible if $m I_Q(X_0;X_2)<mLI_Q(X_1;Y|X_0X_2)$, which is the desired result in~\eqref{eq:block_fading}. This modified constraint is useful in some wireless communication settings for which channel states are often block \ac{iid}. When $L\rightarrow \infty$, which correspond to a single realization of the random state, the information constraint is always satisfied and any $\ol{Q}\in\Delta(\calX)$ is implementable.

Finally, we emphasize that the information constraint obtained when coordinating via actions differs from what would be obtained when coordinating using classical communication~\cite{Shannon1948} with a dedicated channel. If Agent 1 could communicate with Agent 2 through a channel with capacity $C$, then all $\widetilde{Q}\in\Delta(\calX)$ subject to the information constraint
\begin{align}
  I_{\widetilde{Q}}(X_0;X_2)\leq C
\end{align}
would be implementable. In contrast, the constraint $I_{Q}(X_0;X_2)<I_Q(X_1;Y|X_0X_2)$ reflects the following two distinctive characteristics of communication via actions.
\begin{enumerate}
\item The input distribution $X_1$ to the ``implicit channel'' used for communication between Agent 1 and Agent 2 cannot be optimized independently of the actions and of the state.
\item The output $Y$ of the implicit channel depends not only on $X_1$ but also on $(X_0,X_2)$; essentially, the state $X_0$ and the actions $X_2$ of Agent $2$ act as a state for the implicit channel.
\end{enumerate}
Under specific conditions, the coordination via actions may reduce to coordination with a dedicated channel. For instance, if the payoff function factorizes as $w(x_0,x_1,x_2)\eqdef w_1(x_1)w_2(x_0,x_2)$ and if the observation structure satisfies $(X_0,X_2)-X_1-Y$, then any joint distribution $\widetilde{Q}(x_0,x_1,x_2)\eqdef \ddot{Q}(x_0,x_2)\dot{Q}(x_1)$ satisfying the information constraint 
\begin{align}
    I_{\widetilde{Q}}(X_0;X_2) < I_{\dot{Q}}(X_1;Y)
\end{align}
would be \tc{black}{an achievable empirical coordination}; in particular, one may optimize $\dot{Q}$ independently. In addition, if $w_1(x_1)$ is independent of $x_1$, the information constraint further simplifies as
\begin{align}
  I_{\widetilde{Q}}(X_0;X_2) < \max_{\dot{Q}}I_{\dot{Q}}(X_1;Y),
\end{align}
and the implicit communication channel effectively becomes a dedicated channel.\smallskip

\section{Expected payoff optimization}
\label{sec:optimization-problem}

We now study the problem of determining $Q\in\Delta(\calX)$ that leads to the maximal payoff in case I and case II. We establish two formulations of the problem: one that involves $Q$ viewed as a function, and one that explicitly involves the vector of probability masses of $Q$. Although the latter is seemingly more complex, it is better suited to numerically determine the maximum expected payoff and turns out particularly useful in Section~\ref{sec:coded-power-control}. While we study the general optimization problem in Section \ref{sec:general-OP}, we focus on the case of perfect monitoring in Section~\ref{sec:PM}, for which we are able to gain more insight into the structure of the optimal solutions. 

\subsection{General optimization problem}
\label{sec:general-OP}

From the results of Section \ref{sec:information-constraints}, the determination of the largest average payoff requires solving the following optimization problem, with $\ell \in \{\mathrm{I},\mathrm{II} \}$:
\begin{equation}
\begin{array}{clcll}
\text{minimize} & \multicolumn{4}{l}{-\mathbb{E}_{Q}[w(X_0, X_1, X_2)]  =
 - \displaystyle{\sum_{(x_0, x_1, x_2, y, u)} Q(x_0, x_1, x_2, y, u) w(x_0, x_1, x_2)}}\\
\text{s.t.} & \displaystyle{-1 + \sum_{(x_0, x_1, x_2, y, u)}  Q(x_0, x_1, x_2, y,u) }& \stackrel{(c)}{=}& 0 &\\
\forall (x_0, x_1, x_2, y,u)  \in \mc{X} \times \mc{Y} \times \mc{U}, & \frac{\displaystyle{Q(x_0, x_1, x_2, y,u)}}{\displaystyle{\sum_{(y,u)} Q(x_0, x_1, x_2, y,u)}}
-\Gamma(y|x_0, x_1, x_2) &\stackrel{(d)}{=}& 0&  \\
\forall x_0  \in \mc{X}_0 , &\displaystyle{ - \rho_0(x_0) + \sum_{(x_1,x_2,y,u)}  Q(x_0, x_1, x_2, y,u)} &\stackrel{(e)}{=}& 0 &
\\
\forall (x_0, x_1, x_2, y,u)   \in\mc{X} \times \mc{Y} \times \mc{U}, &\displaystyle{-Q(x_0, x_1, x_2, y,u) }&\stackrel{(f)}{\leq}& 0 &\\
 &  \displaystyle{\Phi^{\mathrm{\ell}}(Q)} &\stackrel{(g)}{\leq} & 0 &
\end{array} \label{Optpb-Q-general}
\end{equation}
where in case I $\Phi^{\mathrm{I}}(\cdot)$ is defined in \eqref{eq:def-Phi-I} while in case II
\begin{align}
  \Phi^{\mathrm{II}}(Q) \triangleq I_Q(X_0;X_2) - I_Q(U;Y,X_2) + I_Q(U;X_0,X_2).
\end{align}
We start by addressing the potential convexity of the optimization problem~\cite{Boyd2004}. The objective function to minimize is linear in $Q$ and the constraints $(c)$, $(d)$, $(e)$, and $(f)$ restrict the domain to a convex subset of the unit simplex. Therefore, it suffices to show that the domain resulting from the additional constraint $(g)$ is convex for the optimization problem to be convex. In case I, for which the set $\mc{U}$ reduces to a singleton, \tc{black}{the following lemma proves} that $\Phi^{\mathrm{I}}$ is a convex function of $Q$, which implies that the additional constraint $(g)$ defines a convex domain. 
\begin{lem}\label{lemma:convexity-of-phi}
The function $\Phi^{\mathrm{I}}$ is strictly convex over the set of distributions $Q \in \Delta(\calX \times \mc{Y})$ with marginal $\rho_0\in\Delta(\calX_0)$ that factorize as  
    \begin{equation}
Q(x_0,x_1,x_2,y) = \Gamma(y|x_0,x_1,x_2) \rho_0(x_0) Q(x_1,x_2|x_0),\end{equation} with $\rho_0$ and $\Gamma$ fixed.
\end{lem}
\begin{proof}
  See Appendix~\ref{sec:proof-lemma-refl}.
\end{proof}
For case II, we have not proved that $\Phi^{\mathrm{II}}$ is a convex function but, by using a time-sharing argument, it is always possible to make the domain convex. In the remaining of the paper, we assume this convexification is always performed, so that the optimization problem is again convex. 

We then investigate whether Slater's condition holds, so that the \ac{KKT} conditions become necessary conditions for optimality. Since the problem is convex, the \ac{KKT} conditions would also be sufficient. 

\begin{prop} Slater's condition holds in cases I and II for irreducible
channel transition probabilities \idest such that $\forall
  (x_0,x_1,x_2,y) \in \mc{X}_0 \times \mc{X}_1 \times \mc{X}_2 \times \mc{Y}, \Gamma(y|x_0,x_1,x_2) >0$.
\end{prop}

\begin{IEEEproof}
We establish the existence of a strictly feasible point in case II, from which the existence for case I follows as a special case. Consider a distribution $Q\in\Delta(\calX\times\calY\times\calU)$ such that $X_0$, $X_1$, and $X_2$ are independent, and $U=X_1$. We assume without loss of generality that the support of the marginals $Q_{X_i}$, $i\in\{0,1,2\}$ is full, \idest $\forall x_i \in \mc{X}_i,Q_{X_i}(x_i) >0$. If the channel transition probability is irreducible, note that $Q(x_0,x_1,x_2,y,u)$ is then strictly positive, making the constraint $(f)$ inactive. As for inequality constraint $(g)$, notice that
\begin{align}
I_{Q}(X_0;X_2) - I_{Q}(U;Y,X_2) + I_{Q}(U;X_0,X_2) &= 0 - I_{Q}(X_1;Y) - I(X_1;X_2|Y) + I_{Q}(X_1;X_0,X_2) \\
& = - I_{Q}(X_1;Y)- I_Q(X_1;X_2|Y) \label{Indep} \\
& = -H_Q(X_1) + H_Q(X_1|Y,X_2)\\
& < 0. \label{Mutprop}
\end{align}
Hence, the chosen distribution constitutes a strictly feasible point for the domain defined by constraints $(c)$-$(g)$, and remains a strictly feasible point after convexification of the domain.
\end{IEEEproof}

Our objective is now to rewrite the above optimization problem more explicitly in terms of the vector of probability masses that \tc{black}{describes} $Q$. This is useful not only to exploit standard numerical solvers in Section \ref{sec:coded-power-control}, but also to apply the \ac{KKT} conditions in Section \ref{sec:PM}. We introduce the following notation. Without loss of generality, the finite sets $\calX_k$ for $k\in\{0,1,2\}$ are written in the present section as set of indices $\calX_k= [1; n_k]$; similarly, we write $\mc{U}=[1:n_u]$ and $\calY=[1:n_y]$. With this convention, we define a bijective mapping $\psi^{\ell}:\calX\times\calY \times\mc{U}\rightarrow [1:n^\ell]$ as 
\begin{align}
  \label{eq:indexing_formula}
\psi^{\ell}(i', j', k', l', m')\triangleq  m' + n_u(l'-1)+n_un_y(k'-1)+n_un_yn_2(j'-1)  +n_un_yn_2n_1(i'-1),
\end{align}
which maps a realization $(i', j', k', l', m')\in \calX\times\calY \times\mc{U}$ to a unique index $\psi^{\ell}(i', j', k', l', m')\in [1:n^\ell]$.
We also set $n^{\mathrm{I}} \eqdef n_0n_1n_2n_y$ and $n^{\mathrm{II}} \eqdef n_0n_1n_2n_yn_u$. This allows us to introduce the vector of probability masses $q^{n^{\mathrm{\ell}}}=(q_1,q_2,\dots,q_{n^{\mathrm{\ell}}})$ for $\mathrm{\ell}\in \{\mathrm{I}, \mathrm{II}\}$, in which each component $q_i$, $i \in [1:n^{\ell}]$, is equal to $Q(  (\psi^{\ell})^{-1}(i))$, and the vector of payoff values $w^{n^{\mathrm{\ell}}} = (w_1, w_2, \dots, w_{n^{\mathrm{\ell}}}) \in \mathbb{R}^{n^{\mathrm{\ell}}}$, in which each component $w_i$ is the payoff of $(\psi^{\ell})^{-1}(i)$. The relation between the mapping $Q$ (resp. $w$) and the vector $q^{n^{\mathrm{\ell}}}$ (resp. $w^{n^{\mathrm{\ell}}}$) is summarized in Table \ref{Tab:indexingU}.

\begin{table}[!h]
\caption{Chosen indexation for the payoff vector $w$
and probability distribution vector $q$. Bold lines delineate
blocks of size $n_1 n_2 n_y n_u$ and each block
corresponds to a given value of the random state $X_0$. The $5-$uplets
are sorted according to a lexicographic order.}
\begin{center}
\begin{tabular}{|l||c|c|c|c|c|}
\hline
Index of $q_i$ & $X_0$ & $X_1$ & $X_2$ & $Y$ & $U$\\
\hline\hline
1 & 1&1&1 & 1 & 1\\
2 & 1 & 1 & 1 & 1 & 2 \\
\vdots & \vdots & \vdots & \vdots & \vdots & \\
$n_u$ & 1 & 1 & 1 &  1 & $n_u$\\
\hline
$n_u + 1$ & 1 & 1 & 1 & 2 & 1\\
\vdots & \vdots & \vdots & \vdots & \vdots  & \vdots\\
$2 n_u$ & 1 &  1 & 1 & 2 & $n_u$\\
\hline
\vdots & \vdots & \vdots & \vdots  & \vdots & \vdots \\
\hline
$n_1 n_2 n_y n_u -n_u + 1$ & 1 & $n_1$ & $n_2$ & $n_y$ & 1\\
\vdots & \vdots & \vdots & \vdots & \vdots & \vdots \\
$n_1 n_2 n_y n_u$ & 1 & $n_1$ &$n_2$ & $n_y$ & $n_u$ \\
\hline\hline
\vdots & \vdots & \vdots & \vdots & \vdots & \vdots \\
\vdots & \vdots & \vdots & \vdots & \vdots & \vdots \\
\hline\hline
$(n_0-1) n_1 n_2 n_y n_u + 1$ & $n_0$ & $1$ & $1$ & 1 & $1$\\
\vdots & \vdots & \vdots & \vdots  & \vdots & \vdots \\
$n_0 n_1 n_2 n_y n_u $ & $n_0$ & $n_1$ & $n_2$  & $n_y$ & $n_u$ \\
\hline
\end{tabular}\label{Tab:indexingU}
\end{center}
\end{table}

Using the proposed indexing scheme, the optimization problem is written in standard form as follows.
\begin{equation}
\begin{array}{clc}
\text{minimize} & \multicolumn{2}{l}{-\mathbb{E}_{Q}[w(X_0, X_1, X_2)]  = - \displaystyle{\sum_{i=1}^{n^{\mathrm{\ell}} }} q_i w_i}\\
\text{s.t.} & \displaystyle{-1 + \sum_{i=1}^{n^{\mathrm{\ell}}} q_i}&
 \stackrel{(h)}{=} 0 \\
\forall i \in [1;n^{\mathrm{\ell}}], & \ds{\frac{\displaystyle{q_i}}{\Theta_{i}}} - \Gamma_i  &\stackrel{(i)}{=} 0  \\
\forall i \in [1:n_0], &\displaystyle{ - \rho_{0}(i)  + \sum_{j=1+(i-1)n_1n_2n_yn_u}^{in_1n_2n_yn_u}}
 q_j   &\stackrel{(j)}{=} 0  \\
\forall i \in [1:n^{\mathrm{\ell}}],  &\displaystyle{-q_i }&\stackrel{(k)}{\leq} 0 \\
 &  \displaystyle{\phi^{\mathrm{\ell}}(q^{n^{\mathrm{\ell}}})}
 &\stackrel{(\ell)}{\leq}  0
\end{array} \label{Optpb-Q-gen}
\end{equation}
where
\begin{equation}
\Theta_i \stackrel{\vartriangle}{=} \displaystyle{\sum_{\substack{j\in\{1,\dots,n_yn_u\}\\ k\in\{1,\dots,n_0n_1n_2 \}}} q_{(k-1)n_yn_u+j}. \mathbbm{1}_{\{(k-1)n_yn_u \leq i \leq kn_yn_u -1\}} }
\end{equation}
 and $\forall i \in [1:n_0]$, $\rho_0(i) = \mathrm{P}(X_0=i)$ and $\forall i \in [1:n^{\mathrm{\ell}}]$, $\Gamma_i$ corresponds to the value of $\Gamma(y|x_0,x_1,x_2)$, according to Table \ref{Tab:indexingU}. As for the function associated with inequality constraint $(\ell)$, it writes in case II (case I follows by specialization with $|\mc{U}|=1$) as follows:

\begin{align}
\displaystyle{\phi^{\mathrm{II}}(q^{n^{\mathrm{II}}})} &= I_{q^{n^{\mathrm{II}}}}(X_0;X_2) - I_{q^{n^{\mathrm{II}}}}(U;Y,X_2) + I_{q^{n^{\mathrm{II}}}}(U;X_0,X_2) \nonumber \\
&= H_{q^{n^{\mathrm{II}}}}(X_0)  - H_{q^{n^{\mathrm{II}}}}(U,X_0|X_2) + H_{q^{n^{\mathrm{II}}}}(U|Y,X_2) \nonumber \\
&= H_{q^{n^{\mathrm{II}}}}(X_0) +  H_{q^{n^{\mathrm{II}}}}(X_2) -  H_{q^{n^{\mathrm{II}}}}(X_0,X_2,U) + H_{q^{n^{\mathrm{II}}}}(X_2,Y,U) - H_{q^{n^{\mathrm{II}}}}(X_2,Y)
\end{align}
with
\begin{align}
H_{q^{n^{\mathrm{II}}}}(X_0) = -&
\displaystyle{\sum_{i=1}^{n_0}} \left[ \big(\displaystyle{
\sum_{j=1+(i-1)n_1n_2n_yn_u}^{in_1n_2n_yn_u}} q_j\big)\log\big(\displaystyle{\sum_{j=1+(i-1)n_1n_2n_yn_u}^{i n_1n_2n_yn_u}} q_j\big) \right],
\end{align}

\begin{align}
H_{q^{n^{\mathrm{II}}}}(X_2) = -&
\displaystyle{\sum_{i=1}^{n_2}} \left[ \big(\displaystyle{\sum_{j=1}^{n_0n_1}}
\sum_{k=1}^{n_yn_u} q_{(i-1)n_un_y + (j-1)n_2n_yn_u + k}\big) \log\big(\displaystyle{\sum_{j=1}^{n_0n_1}}\sum_{k=1}^{n_yn_u} q_{(i-1)n_un_y + (j-1)n_2n_yn_u + k}\big)  \right],
\end{align}

\begin{align}
H_{q^{n^{\mathrm{II}}}}(X_2,Y,U) = -& \displaystyle{\sum_{i=1}^{n_2n_yn_u}} \left[ \big(\displaystyle{\sum_{j=1}^{n_0n_1}} q_{(j-1)n_2n_yn_u + i}\big)\log\big(\displaystyle{\sum_{j=1}^{n_0n_1}} q_{(j-1)n_2n_yn_u + i}\big)
\right],
\end{align}

\begin{align}
&H_{q^{n^{\mathrm{II}}}}(X_0,X_2,U) = - \displaystyle{\sum_{i=1}^{n_0}}
\displaystyle{\sum_{j=1}^{n_2}}  \displaystyle{\sum_{k=1}^{n_u}}  \Bigg[ \big(\displaystyle{\sum_{l=1}^{n_0n_1}} \displaystyle{\sum_{m=1}^{n_y}} q_{(i-1)n_1n_2n_yn_u+(j-1)n_yn_u+k+(l-1)n_2n_yn_u+(m-1)n_u}\big) \nonumber\\
&\phantom{=========} \log\big(\displaystyle{\sum_{l=1}^{n_0n_1}} \displaystyle{\sum_{m=1}^{n_y}} q_{(i-1)n_1n_2n_yn_u+(j-1)n_yn_u+k+(l-1)n_2n_yn_u+(m-1)n_u}\big) \Bigg],
\end{align}
and
\begin{align}
H_{q^{n^{\mathrm{II}}}}(X_2,Y) = -& \displaystyle{\sum_{i=1}^{n_2n_y}} \left[ (\displaystyle{\sum_{j=1}^{n_0n_1}}\displaystyle{\sum_{k=1}^{n_u}} q_{(j-1)n_2n_yn_u + (i-1)n_u+k}) \log(\displaystyle{\sum_{j=1}^{n_0n_1}}\displaystyle{\sum_{k=1}^{n_u}} q_{(j-1)n_2n_yn_u + (i-1)n_u+k}) \right].
\end{align}

This formulation is directly exploited in Section~\ref{sec:PM} and in Section \ref{sec:coded-power-control}.

\subsection{Optimization problem for perfect monitoring}
\label{sec:PM}

In the case of perfect monitoring, for which Agent $2$ perfectly monitors Agent $1$'s actions and $Y=X_1$,
 the information constraints \eqref{eq:lower-bound-I} and \eqref{eq:lower-bound-II} coincide and 
\begin{align}
 \phi(q^n) &\stackrel{\vartriangle}{=} \displaystyle{\phi^{\mathrm{I}}(q^{n^{\mathrm{I}}})}
 = \displaystyle{\phi^{\mathrm{II}}(q^{n^{\mathrm{II}}})} = H_{q^{n}}(X_2) -  H_{q^{n}}(X_2|X_0) -  H_{q^{n}}(X_1|X_0,X_2)
\end{align}
with $q^n =(q_1,\dots,q_n)$, $n=n_0n_1n_2$. To further analyze the relationship between the vector of payoff values $w^n$ and an optimal joint distribution $q^n$, we explicitly express the \ac{KKT} conditions. The Lagrangian is
\begin{multline}
 \mathcal{L}(q^n,\lambda^n,\mu_0, \mu^{n_0}, \lambda_{\mathrm{IC}}) =
 - \sum_{i=1}^{n} w_i q_i + \lambda_i q_i + \mu_0\left[-1+\sum_{i=1}^{n} q_i\right] + \sum_{j=1}^{n_0} \mu_{j} \left[- \rho_{0i} + \sum_{i=1+(j-1)n_1n_2}^{j n_1 n_2} q_i \right] \\ 
 + \lambda_{\mathrm{IC}} \phi(q^n)
\end{multline}
where $\lambda^n = (\lambda_1, \dots, \lambda_{n})$, $\mu^{n_0} = (\mu_1, \dots, \mu_{n_0})$, and the subscript IC stands for information constraint.

A necessary and sufficient condition for a distribution $q^n$ to be an optimum point is that it is a solution of the
following system:
\begin{align}
& \forall \; i \in[1:n], \ \frac{\partial \mathcal{L}}{\partial q_i} = -w_i - \lambda_i + \mu_0 + \sum_{j=1}^{n_0} \mu_{j} \mathbbm{1}_{ \{1+n_1 n_2(j-1) \leq i \leq j n_1 n_2\} }  + \lambda_{\mathrm{IC}} \frac{\partial \phi}{\partial q_i}(q^n) = 0 \qquad 
\label{partialLAGg} \\
& q^n \text{ verifies }  (h), (i), (j)\\
&\forall \; i \in[1:n], \ \lambda_i \geq 0  \\
&\lambda_{\mathrm{IC}} \geq 0 \\
& \forall \; i \in[1:n], \  \lambda_i q_i = 0\\
&\lambda_{\mathrm{IC}} \phi(q^n) = 0
\end{align}
where
\begin{multline}
\forall i \in [1:n], \ \frac{\partial \phi}{\partial q_i}(q^n) =
\Bigg[-\sum_{k=1}^{n_0}\bigg( \mathbbm{1}_{\{1+(k-1)n_1n_2 \leq i \leq kn_1n_2\} } \log\sum_{j=1+(k-1)n_1n_2}^{k n_1n_2} q_j\bigg)\\
 -\sum_{k=1}^{n_2} \mathbbm{1}_{\{ i \in \{k,k+n_2,\dots,k+(n_0n_1-1)n_2 \} \} } \log\sum_{j=0}^{n_0n_1-1} q_{k+jn_2} -1 + \log q_i  \Bigg].
\end{multline}

In the following, we assume that there exists a permutation of $[1:n]$ such that the vector of payoff values $w^{n}$ after permutation of the components is strictly ordered. A couple of observations can then be made by inspecting the \ac{KKT} conditions above. First, if the expected payoff were only maximized under the constraints $(h)$ and $(k)$, the best joint distribution would be to only assign probability to the greatest element of the vector $w^n$; 
in other words the best $q^n$ would correspond to a vertex of the unit simplex $\Delta(\calX)$. However, as the distribution of the random state fixed by constraint $(j)$, at least $n_0$ components of $q^n$ have to be positive. It is readily verified that under constraints $(h),(j),$ and $(k)$, the optimal solution is that for each $x_0$ the optimal pair $(x_1,x_2)$ is chosen; therefore, $q^n$ possesses exactly $n_0$ positive components. This corresponds to the costless communication scenario. Now, in the presence of the additional information constraint $(\ell)$, the optimal solutions contain in general more than $n_0$ positive components because optimal communication between the two agents requires several symbols of $\mc{X}_1$ to be associated with a given realization of the state. In fact, as shown in the following proposition, there is a unique optimal solution under \tc{black}{mild assumptions}.

\begin{prop}\label{prop:uniq-strict-order} If there exists a permutation such that the payoff vector $w^n$ is strictly ordered, then the optimization problem (\ref{Optpb-Q-gen}) has a unique solution.  
\end{prop}

\begin{IEEEproof}
Assume $\lambda_{\mathrm{IC}}=0$ in the Lagrangian. Further assume that a candidate solution of the optimization problem $q^n$ has two or more positive components in a block of size $n_1 n_2$ associated with a given realization $x_0$ (see Table \ref{Tab:indexingU}). Then, there exist two indices $(i_1, i_2)$ such that $\lambda_{i_1}=0$ and $\lambda_{i_2}=0$. Consequently, the conditions on the gradient $\frac{\partial \mc{L}}{\partial q_i} = 0$ for $i\in\{i_1,i_2\}$ imply that $w_{i_1} = w_{i_2}$, which contradicts the assumption of $w^n$ being strictly ordered under permutation. Therefore, a candidate solution only possesses a single positive component per block associated with a given realization $x_0$, which means that $X_1$ and $X_2$ are deterministic functions of $X_0$. Hence, $H_{q^n}(X_2|X_0) = H_{q^n}(X_1|X_0X_2)=0$ and the information constraint reads $H_{q^n}(X_2) <0$, which is impossible. Hence, $\lambda_{\mathrm{IC}}>0$.

From Lemma \ref{lemma:convexity-of-phi}, we know that $\phi(q^n)$ is strictly convex. Since $\lambda_{\mathrm{IC}}>0$, the Lagrangian is the sum of linear functions and a strictly convex function. Since it is also continuous and optimized over a compact and convex set, there exists a maximum point and it is unique.
\end{IEEEproof}

Apart from assuming that $w^n$ can be strictly ordered, Proposition \ref{prop:uniq-strict-order} does not assume anything on the values of the components of $w^n$. In practice, for a specific problem it will be relevant to exploit the special features of the problem of interest to better characterize the relationship between the payoff function (which is represented by $w^n$) and the optimal joint probability distributions (which are represented by the vector $q^n$). This is one of the purposes of the next section.

\section{Coded power control}
\label{sec:coded-power-control}

We now exploit the framework previously developed to study power control in interference networks. In this context, the agents are the transmitters and the random state corresponds to the global wireless channel state, \idest all the channel gains associated with the different links between transmitters and receivers. \emph{\ac{CPC}} consists in embedding information about the global wireless channel state into transmit power levels themselves rather than using a dedicated signaling channel. Provided that the power levels of a given transmitter can be observed by the other transmitters, the sequence of power levels can be used to coordinate with the other transmitters. Typical mechanisms through which agents may observe power levels include sensing, as in cognitive radio settings, or feedback, as often assumed in interference networks. One of the salient features of coded power control is that interference is directly managed in the radio-frequency domain and does not require baseband detection or decoding, which is useful in systems such as heterogeneous networks. 
The main goal of this section is to assess the limiting performance of coded power control and its potential performance gains over other approaches, such as the Nash equilibrium power control policies of a given single-stage non-cooperative game. This comparison is relevant since conventional distributed power control algorithms, such as the iterative water-filling algorithm, do not exploit the opportunity to exchange information through power levels or vectors to implement a better solution\tc{black}{, \eg that would Pareto-dominate the Nash equilibrium power control policies}.

\subsection{Coded power control over interference channels}

We first consider an interference channel with two transmitters and two receivers, which we then specialize to the multiple-access channel in Section~\ref{sec:coord-scheme} to develop and analyze an explicit non-trivial power control code. 

By denoting $g_{ij}$ the channel gain between Transmitter $i$ and Receiver $j$, each realization of the global wireless channel state is given by
\begin{equation}
\label{eq:globa_channel_state}
x_0=(g_{11},g_{12},g_{21},g_{22}),
\end{equation}
where $g_{ij}\in \mc{G}$, $|\mc{G}|<\infty$; it is further assumed that the channel gains $g_{ij}$ are independent and we set $\mc{X}_0=\mc{G}^4$. Each alphabet $\mc{X}_i, |\mc{X}_i|<\infty$, $i\in\{1,2\}$, represents the set of possible power levels for Transmitter $i$. Assuming that the sets are discrete is of practical interest, as there exist wireless communication standards in which the power can only be decreased or increased by step and in which quantized wireless channel state information is used. In addition, the use of discrete power levels may not induce any loss of optimality~\cite{gjendemsj-twc-2008} w.r.t. the continuous case, as further discussed in Section~\ref{sec:influence-payoff}. 

We consider three stage payoff functions $w^{\mathrm{rate}}$, $w^{\mathrm{SINR}}$, and $w^{\mathrm{energy}}$, which respectively represent the sum-rate, the sum-\ac{SINR}, and the sum-energy efficiency. Specifically,
\begin{equation}
\begin{array}{cccc}
w^{\mathrm{rate}}: &
\calX
& \rightarrow & \mathbb{R}_+\\
& (x_0,x_1,x_2)  & \mapsto & \ds{\sum_{i=1}^2} \log_2 \bigg(1+
\underbrace{\frac{g_{ii} x_{i}}{\sigma^2 + g_{-ii} x_{-i}}}_{\mathrm{SINR}_i} \bigg)\label{eq:SINR_payoff}
\end{array},
\end{equation}
\begin{equation}
\begin{array}{cccc}
w^{\mathrm{SINR}}: &  \calX
& \rightarrow & \mathbb{R}_+\\
& (x_0,x_1,x_2)  & \mapsto & \ds{\sum_{i=1}^2}
\frac{g_{ii} x_{i}}{\sigma^2 + g_{-ii} x_{-i}}
\end{array},
\end{equation}
\begin{equation}
\label{eq:energy-payoff}
\begin{array}{cccc}
w^{\mathrm{energy}}: & \calX
& \rightarrow & \mathbb{R}_+\\
& (x_0,x_1,x_2)  & \mapsto & \ds{\sum_{i=1}^2} \frac{F\left(1+
\frac{g_{ii} x_{i}}{\sigma^2 + g_{-ii} x_{-i}} \right)}{x_i}
\end{array}.
\end{equation}
The notation $-i$ stands for the transmitter other than $i$; $\sigma^2$ corresponds to the reception noise level; $F:\mathbb{R}_+ \rightarrow [0,1]$ is a sigmoidal and increasing function that typically represents the block success rate, see \eg~\cite{meshkati-spmag-2007,belmega-tsp-2011}. The function $F$ is chosen so that $w^{\mathrm{energy}}$ is continuous and has a limit when $x_i\rightarrow 0$. The motivation for choosing these three payoff functions is as follows.
\begin{itemize}
\item The sum-rate is a common measure of performance for distributed power control in wireless networks.
\item The sum-SINR is not only a linear approximation of the sum-rate but also an instance of sum-payoff function that is more sensitive to coordination, since the dependency with respect to the SINR is linear and not logarithmic.
\item The sum-energy efficiency has recently gathered more attention as a way to study a tradeoff between the transmission benefit (namely, the net data rate which is represented by the numerator of the individual payoff) and the transmission cost (namely, the transmit power which is represented by the denominator of the individual payoff). As pointed out in \cite{varma-tvt-2015}, energy-efficiency can even represent the energy consumed in a context with packet re-transmissions, indicating its relevance for green communications. Finally, as simulations reveal next, total energy-efficiency may be very sensitive to coordination.
\end{itemize}

Finally, we consider the following three possible observation structures.
\begin{itemize}
  \item \emph{Perfect monitoring}, in which Agent $2$ directly observes the actions of Agent $1$, \idest $Y =X_1$;
  \item \emph{\acs{BSC} monitoring}, in which Agent $2$ observes the actions of Agent $1$ through a \ac{BSC}. \tc{black}{The channel is given by the alphabets $\mathcal{X}_1 = \{P_{\min}, P_{\max} \} $,  $\mathcal{Y} = \{P_{\min}, P_{\max} \} $, and the transition probability: $ \mathrm{P}(Y=y | X_1= x_1)  = 1 -p $ if $y = x_1$ and $ \mathrm{P}(Y=y | X_1 = x_1)  = p $ if $y \neq x_1$, $0 \leq p \leq 1$};
  \item \emph{Noisy \ac{SINR} feedback monitoring}, in which Agent $2$ observes a noisy version of the \ac{SINR} of Agent $1$ as illustrated in Fig.~\ref{fig:SINR-feedback}; this corresponds to a scenario in which a feedback channel exists between Receiver $2$ and Transmitter $2$. \tc{black}{The channel is given by the alphabets $\mathcal{X}_0 = \mathcal{G}^4$, $\mathcal{X}_1 = \{0, P_{\max} \}$, $\mathcal{X}_2 = \{0, P_{\max} \}$, $\mathcal{Y} = \{\mathrm{SINR}^{(1)}, ..., \mathrm{SINR}^{(N)} \}$, and the transition probability $\mathrm{P}(Y= \mathrm{SINR}^{(n)} |  (X_0,X_1,X_2) = (x_0,x_1,x_2) )  =   \mathrm{P}(Y= \mathrm{SINR}^{(n)} | Y_0 = \mathrm{SINR}^{(m)}) \delta_{ \mathrm{SINR}^{(m)} - \gamma(x_0,x_1,x_2) } $, where {$\gamma$ is the function given by the SINR definition{~\eqref{eq:SINR_payoff}} and
        \begin{align*}
           \mathrm{P}(Y= \mathrm{SINR}^{(n)} | Y_0 = \mathrm{SINR}^{(m)})=\left\{
          \begin{array}{l}
            1-e\text{ if $m=n$,}\\
            e \text{ if $m=1$ and $n=2$, or $m=N$ and $n=N-1$,}\\
            \frac{e}{2}\text{ else.}
          \end{array}
          \right.
        \end{align*}
      }}
\end{itemize}

\begin{figure}[!ht]
\begin{center}
\includegraphics{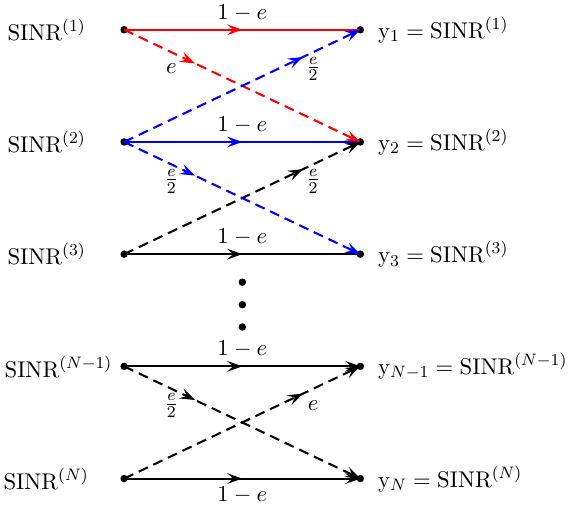}
\caption{The signal observed by Transmitter $2$ lies in an $N-$symbol alphabet
 \idest $|\mc{Y}|=N$. The symbols correspond to possible values for
 the \ac{SINR} at the receiver
 which is associated with Transmitter $2$. Transmitter $2$ observes the
 actual value of its \ac{SINR} with probability $1-e$ while there is probability
 $e\geq0$ that a neighboring \ac{SINR} is observed. In the simulations,
 $e=0.1$.}
\label{fig:SINR-feedback}
\end{center}
\end{figure}

The performance of coded power control will be assessed against that of the following three benchmark power control policies.
\begin{itemize}
\item \emph{\acf{NPC} policy}. In such a policy, each transmitter aims at maximizing an individual stage payoff function $u_i(x_0,x_1,x_2)$. In the sum-rate, sum-\ac{SINR}, and sum-energy efficiency cases, these individual stage payoff-functions are respectively given by $u_i(x_0,x_1,x_2) = \log_2(1+ \mathrm{SINR}_i)$, $u_i(x_0,x_1,x_2) = \mathrm{SINR}_i$, and $u_i(x_0,x_1,x_2) = \frac{F(\mathrm{SINR}_i)}{x_i}$. In the sum-rate and sum-\ac{SINR} cases, the unique Nash equilibrium is $(x_1^{\mathrm{NE}},x_2^{\mathrm{NE}}) =(P_{\max}, P_{\max})$, irrespectively of the value of $x_0$, $P_{\max}$ being the maximal power level for the transmitters. In the sum-energy efficiency case, the unique non-trivial Nash equilibrium may be determined numerically and generally requires some knowledge of $x_0$, depending on how it is implemented (see e.g., \cite{lasaulce-book-2011}).
\item \emph{\acf{SPC} policy}. This policy corresponds to a basic coordination scheme in which Transmitter $1$ optimizes its power knowing that Transmitter $2$ transmits at full power; \ac{SPC} requires the knowledge of the current wireless channel state realization at Transmitter $1$. Specifically, $x_2= P_{\max}$, $x_1^{\dag} \in \ds{\arg \max_{x_1}} \;w^{\text{r}}(x_0, x_1, P_{\max} )$, $\mathrm{r}\in \{\mathrm{rate}, \mathrm{SINR}, \mathrm{energy} \}$. \tc{black}{SPC is a rather intuitive scheme, which also corresponds to the situation in which Transmitter $1$ only knows the past and current realizations of the state; this scheme can in fact be optimal if Agent $1$ and Agent 2 only know $X_0^{n}$ and $Y^{n-1}$, respectively, with Agent $2$ choosing the best constant action. Therefore, comparisons with SPC allow us to assess the potential gain of an advanced coding scheme and the value of knowing the future.} 

 \item \emph{\acf{CCPC} policy.} This policy corresponds to the situation in which transmitters may communicate at not cost, so that they may jointly optimize their powers to achieve the maximum of the payoff function $w^{\text{r}}$ at every stage. In such a case there is no information constraint, and the performance of \ac{CCPC} provides an upper bound for the performance of all other policies. 
\end{itemize}

The communication \ac{SNR} is defined as
\begin{equation}
\text{SNR(dB)} \triangleq 10\log_{10}\frac{P_{\max}}{\sigma^2}.
\end{equation}

\subsection{Influence of the payoff function}
\label{sec:influence-payoff}

The objective of this subsection is to numerically assess the relative performance gain of \ac{CPC} over \ac{SPC} in the case of perfect monitoring. We assume that the channel gains $g_{ij}\in\{g_{\min}, g_{\max}\}$ are Bernoulli distributed $g_{ij} \sim \mc{B}(p_{ij})$ with $p_{ij}\triangleq \mathrm{P}(g_{ij} = g_{\min})$; with our definition of $X_0$ in~\eqref{eq:globa_channel_state}, this implies that $|\mc{X}_0|=16$. All numerical results in this subsection are obtained for $g_{\min} = 0.1$, $g_{\max}=2$ and $(p_{11},p_{12},p_{21},p_{22}) = (0.5,0.1,0.1,0.5)$. The sets of transmit powers $\mc{X}_1$, $\mc{X}_2$ are both assumed to be the same alphabet of size four $\{P_1,P_2,P_3,P_4 \}$, with $P_1=0, P_2=\frac{P_{\max}}{3},   P_3=\frac{2P_{\max}}{3}, P_4 = P_{\max}$. The quantity $P_{\max}$ is given by the operating SNR and $\sigma^2=1$. The function $F$ is chosen as a typical instance of the efficiency function used in~\cite{belmega-tsp-2011}, \idest
   \begin{equation}
   F(x) = \exp \left(-\frac{2^{0.9}-1}{x}\right).
   \end{equation}
For all $ \mathrm{r}\in \{\mathrm{rate}, \mathrm{SINR}, \mathrm{energy} \}$, the relative performance gain with respect to the \ac{SPC} policy is
\begin{equation}
\text{Relative gain (\%)} =
\left(\frac{ \mathbb{E}_{\ol{Q}^\star}(w^{\mathrm{r}})}{\mathbb{E}_{\rho_0}( \ds{\max_{x_1}}
 \,   w(x_0, x_1, P_{\max} ))  } -1\right) \times 100
\end{equation}
where $\ol{Q}^\star$ is obtained by solving optimization problem (\ref{Optpb-Q-gen}) under perfect monitoring. This optimization is numerically performed using the Matlab function {\small \textsf{fmincon}}. Fig.~\ref{fig92b} illustrates the relative performance gain in \% \ac{wrt} the \ac{SPC} policy for the sum-energy efficiency, while Fig. \ref{fig92c} illustrates it for the sum-\ac{SINR} and sum-rate. 

As shown in Fig.~\ref{fig92b}, our simulation results suggest that \ac{CPC} provides significant performance gains for the sum-energy efficiency. This may not be surprising, as the payoff function~\eqref{eq:energy-payoff} is particularly sensitive to the lack of coordination; in fact, as the transmit power becomes high, $\frac{F(\mathrm{SINR}_i)}{x_i} \rightarrow \frac{1}{x_i}$, which means that energy efficiency decreases rapidly. As shown in Fig.~\ref{fig92c}, the performance gains of \ac{CPC} for the sum-\ac{SINR} and the sum-rate are more moderate, with gains as high as 43\% for the sum-\ac{SINR} and 25\% for the sum-rate; nevertheless, such gains are still significant, and would be larger if we used \ac{NPC} instead of \ac{SPC} as the reference case, as often done in the literature of distributed power control. \tc{black}{The shape of the sum-rate curve in Fig.~\ref{fig92c} can be explained intuitively. At low SNR, interference is negligible and the sum-rate is maximized when both transmitters use full power, which is also what SPC does in this regime. At high SNR, SPC is not optimal but still provides a large sum-rate, which is comparable to that provided by the best CPC scheme. Between these regimes, advanced coordination schemes are particularly useful, which explains the peak at intermediate SNR.}

\begin{figure}[htbp]
\begin{center}
  \includegraphics[width=0.70\textwidth]{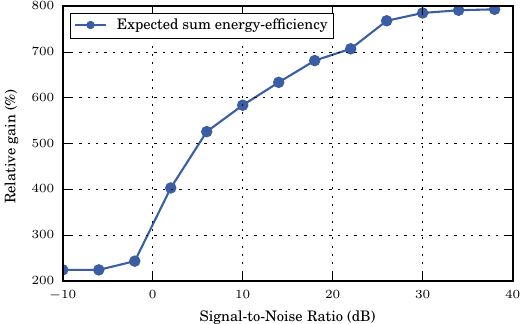}
  \end{center}
\caption{Relative sum-energy gain of coded power control (CPC) with perfect monitoring over semi-coordinated power control (SPC).}
\label{fig92b}
\end{figure}

\begin{figure}[htbp]
\begin{center}
  \includegraphics[width=0.70\textwidth]{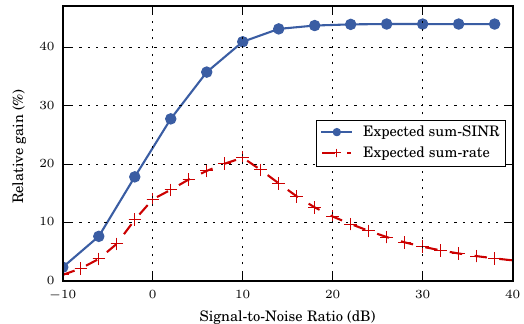}
  \end{center}
\caption{Relative sum-\ac{SINR} gain and sum-rate gain of coded power control (CPC) with perfect monitoring over semi-coordinated power control (SPC).}
\label{fig92c}
\end{figure}

We conclude this subsection by providing the marginals $\ol{Q}_{X_1}^{\star}(x_1)=\ds{\sum_{x_0,x_2}} \ol{Q}^{\star}(x_0,x_1,x_2)$, $\ol{Q}_{X_2}^{\star}(x_2)=\ds{\sum_{x_0,x_1}}
 \ol{Q}^{\star}(x_0,x_1,x_2)$, and joint distribution $\ol{Q}_{X_1 X_2}^{\star}(x_1,x_2)=\ds{\sum_{x_0}} \ol{Q}^{\star}(x_0,x_1,x_2)$ of the optimal joint distribution for \ac{CPC} and \ac{CCPC} in Table~\ref{Tabdistributions-WO-IC} and Table~\ref{Tabdistributions}, respectively. In both cases, the results correspond to the maximization of the sum-rate payoff function $w^{\mathrm{rate}}$ and $\mathrm{SNR}=10$ dB. Table~\ref{Tabdistributions-WO-IC} shows that, without information constraint, the sum-rate is maximized when the transmitters  correlate their power levels so that only three pairs of transmit power levels are used out of $16$. This result is consistent with~\cite{gjendemsj-twc-2008}, which proves that, for interference channels with two transmitter-receiver pairs, there is no loss of optimality in terms of $w^{\mathrm{rate}}$ by operating over a binary set  $\{0,P_{\max}\}$ instead of a continuous interval $[0, P_{\max}]$. Interestingly, as seen in Table \ref{Tabdistributions}, the three best configurations of the \ac{CCPC} policy are exploited $44.3+42.9+2.1=89.3\%$ of the time in the \ac{CPC} policy, despite the presence of communication constraints between the two transmitters.

\begin{table*}[!t]
\caption{Optimal marginal and joint distributions (expressed in \%) for the sum-rate payoff function of the \ac{CCPC} policy, with $\mathrm{SNR}=10$ dB and
with four possible transmit power levels $\left\{0,\frac{10}{3}, \frac{20}{3},10 \right\}$.}
\label{Tabdistributions-WO-IC}
\begin{equation*}
{\normalsize
\begin{tabular}{|c|c|c|c|c|}
\hline
$(\tc{red}{\ol{Q}_{X_1}^{\star}(x_1)}, \ol{Q}_{X_2}^{\star}(x_2), {\bf \ol{Q}_{X_1 X_2}^{\star}(x_1,x_2)}) $ &
 $x_1= 0$ & $x_1= \frac{10}{3}$
&  $x_1 = \frac{20}{3}$ & $x_1=10$ \\
 in \%   & & & & \\
\hline
$x_2= 0\tc{white}{0}$ & (\tc{red}{47.5},47.5,{\bf00.0}) & (\tc{red}{00.0},47.5,{\bf00.0}) &
(\tc{red}{00.0},47.5,{\bf00.0}) &
 (\tc{red}{52.5},47.5,{\bf47.5})\\
\hline
$x_2= \frac{10}{3}$ & (\tc{red}{47.5},00.0,{\bf00.0}) &
(\tc{red}{00.0},00.0,{\bf00.0}) & (\tc{red}{00.0},00.0,{\bf00.0}) &
(\tc{red}{52.5},00.0,{\bf00.0})\\
 \hline
$x_2 = \frac{20}{3}$ & (\tc{red}{47.5},00.0,{\bf00.0}) & (\tc{red}{00.0},00.0,{\bf00.0}) &
 (\tc{red}{00.0},00.0,{\bf00.0})
& (\tc{red}{52.5},00.0,{\bf00.0})\\
 \hline
 $x_2 = 10$ & (\tc{red}{47.5},52.5,{\bf47.5}) &
 (\tc{red}{00.0},52.5,{\bf00.0}) & (\tc{red}{00.0},52.5,{\bf00.0}) &
 (\tc{red}{52.5},52.5,\textbf{05.5}) \\
 \hline
\end{tabular}}
\end{equation*}
\end{table*}

\begin{table*}[!t]
\caption{Optimal marginal and joint distributions (expressed in \%) for the sum-rate payoff function of the \ac{CPC} policy, with $\mathrm{SNR}=10$ dB and
with four possible transmit power levels $\left\{0,\frac{10}{3}, \frac{20}{3},10 \right\}$.}
\label{Tabdistributions}
\begin{equation*}
{\normalsize
\begin{tabular}{|c|c|c|c|c|}
\hline
$(\tc{black}{\ol{Q}_{X_1}^{\star}(x_1)}, \ol{Q}_{X_2}^{\star}(x_2), {\bf \ol{Q}_{X_1 X_2}^{\star}(x_1,x_2)}) $ &
 $x_1= 0$ & $x_1= \frac{10}{3}$
&  $x_1 = \frac{20}{3}$ & $x_1=10$ \\
 in \%   & & & & \\
\hline
$x_2= 0\tc{white}{0}$ & (\tc{black}{44.4},50.4,{\bf00.1}) & (\tc{black}{02.6},50.4,{\bf00.9}) &
(\tc{black}{08.0},50.4,{\bf06.5}) &
 (\tc{black}{45.0},50.4,{\bf42.9})\\
\hline
$x_2= \frac{10}{3}$ & (\tc{black}{44.4},00.0,{\bf00.0}) &
(\tc{black}{02.6},00.0,{\bf00.0}) & (\tc{black}{08.0},00.0,{\bf00.0}) &
(\tc{black}{45.0},00.0,{\bf00.0})\\
 \hline
$x_2 = \frac{20}{3}$ & (\tc{black}{44.4},00.0,{\bf00.0}) & (\tc{black}{02.6},00.0,{\bf00.0}) &
 (\tc{black}{08.0},00.0,{\bf00.0})
& (\tc{black}{45.0},00.0,{\bf00.0})\\
 \hline
 $x_2 = 10$ & (\tc{black}{44.4},49.6,{\bf44.3}) &
 (\tc{black}{02.6},49.6,{\bf01.7}) & (\tc{black}{08.0},49.6,{\bf01.5}) &
 (\tc{black}{45.0},49.6,\textbf{02.1}) \\
 \hline
\end{tabular}}
\end{equation*}

\end{table*}


\subsection{Influence of the observation structure}
\label{sec:influence-obs-struct}

In this subsection, we focus on the observation structure defined by case I in~\eqref{eq:strategies-I} and we restrict our attention to the sum-rate payoff function $w^{\mathrm{rate}}$. The set of powers is restricted to a binary set $\mc{X}_1=\mc{X}_2=\{0, P_{\max}\}$, but unlike the study in Section \ref{sec:influence-payoff}, we do not limit ourselves to perfect monitoring. Fig. \ref{fig103} shows the relative performance gain \ac{wrt} the \ac{SPC} policy as a function of \ac{SNR} for three different observation structures. The performance of \ac{CPC} for \ac{BSC} monitoring is obtained assuming a probability of error of $5\%$, \idest $Z_1\sim\mc{B}(0.05)$, $\mathrm{P}(Z_1=1)=0.05$. The performance of \ac{CPC} for noisy \ac{SINR} feedback monitoring is obtained assuming $e=0.1$; in this case, it can be checked that the \ac{SINR} can take one of $N=7$ distinct values. 

Fig.~\ref{fig103} suggests that \ac{CPC} provides a significant performance gain over \ac{SPC} over a wide range of operating SNRs irrespective of the observation structure. Interestingly, for \ac{SNR} $=10$ dB, the relative gain of \ac{CPC} only drops from $22\%$ with perfect monitoring to $18\%$ with BSC monitoring, which suggest that for observation structures with typical noise levels the benefits of \ac{CPC} are somewhat robust to observation noise. Similar observations can be made for \ac{SINR} feedback monitoring. Note again that one would obtain higher performance gains by considering \ac{NPC} as the reference policy or by considering scenarios with stronger interference. 

\begin{figure}[htbp]
\hspace{-1.2cm}
\begin{center}
\includegraphics[width=0.70\textwidth]{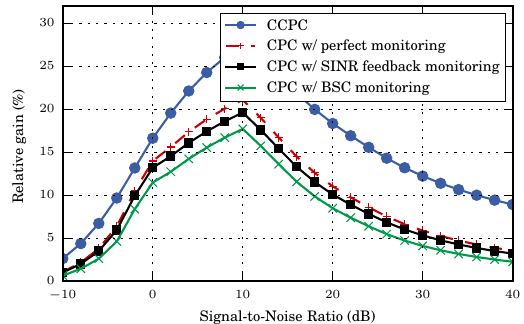}
\end{center}
\caption{Relative sum-rate gain of costless communication power control (CCPC) and coded power control (CPC) over semi-coordinated power control (SPC) under various monitoring assumptions in the observation structure of Case I.}
\label{fig103}
\end{figure}

\subsection{Influence of the wireless channel state knowledge}
\label{sec:influence-X0}
In this subsection, we restrict our attention to \ac{CPC} with BSC monitoring with the same parameters as in Section \ref{sec:influence-obs-struct}, but we consider both Case I and
Case II defined in \eqref{eq:strategies-I} and~\eqref{eq:strategies-II}, respectively. The results for Case II are obtained assuming that $|\mc{U}|=10$. While we already know that the performance of \ac{CPC} is the same in Case I and Case II with perfect monitoring, the results in Fig.~\ref{fig104} suggest that, for typical values of the observation noise, not knowing the past realizations of the global wireless channel state at Transmitter 2 only induces a small performance loss. 

\begin{figure}[htbp]
  \centering
  \includegraphics[width=0.70\textwidth]{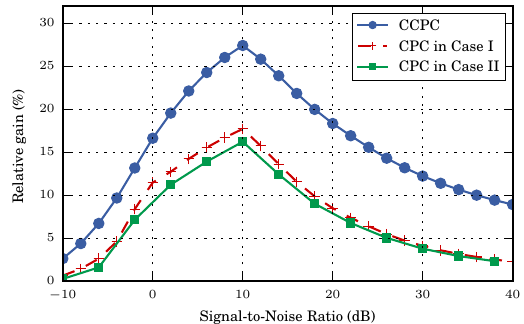}
\caption{Relative sum-rate gain of coded power control (CPC) and costless communication power control (CCPC) over semi-coordinated power control (SPC) for the binary symmetric channel (BSC) monitoring in the observation structure of Case I and Case II.}
\label{fig104}
\end{figure}

\subsection{Influence of the coordination scheme}
\label{sec:coord-scheme}

In this last subsection, we assess the benefits of \ac{CPC} for an explicit code that operates over blocks of length $n=3$. To simplify the analysis and clarify the interpretation, several 
assumptions are made. First, we consider a multiple-access channel, which is a special case of the interference channel studied earlier with two transmitters and a single receiver, so that the global wireless channel state comprises only two components $(g_1,g_2)$. Second, we assume that the global wireless channel state $X_0$ takes values in the binary alphabet $\mc{X}_0\in\left\{(g_{\min}, g_{\max}), (g_{\max}, g_{\min}) \right\}$, and is distributed according to Bernoulli random variable $\mc{B}(p)$ with $p = \mathrm{P}\big(X_0 =(g_{\min}, g_{\max})\big)$. In the remaining of this subsection, we identify the realization $(g_{\min}, g_{\max})$ with ``0'' and $(g_{\max}, g_{\min})$ wth ``1,'' so that we may write $\calX_0=\{0,1\}$. Third, we assume that the transmitters may only choose power values in $\{P_{\mathrm{min}},P_{\mathrm{max}}\}$, and we identify power $P_{\mathrm{min}}$ with ``0'' and power $P_{\mathrm{max}}$ with ``1'', so that we may also write $\calX_1=\calX_2=\{0,1\}$. Finally, we consider the case of perfect monitoring and we restrict our attention to the sum-\ac{SINR} payoff function $w^{\mathrm{SINR}}$. 

The values of the payoff function used in numerical simulations are provided in Fig. \ref{fig-PCGame} as the entries in a matrix. Each matrix corresponds to a different choice of the wireless channel state $x_0$; in each matrix, the choice of the row corresponds to the action $x_1$ of Transmitter $1$, which the choice of the column corresponds to the action $x_2$ of Transmitter $2$.

\begin{figure}[htbp]
\begin{center}
\includegraphics{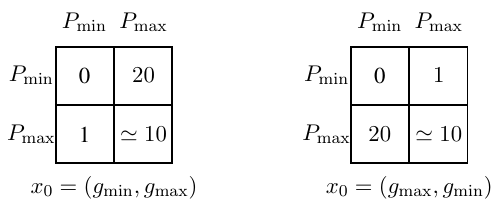}
 \caption{Payoff matrix of $w^{\textrm{SINR}}$ for power control over multiple-access channel. Numerical values of the payoff correspond  to $g_{\mathrm{min}}=0.1$, $g_{\mathrm{max}}=2$, $\sigma^2=1$, $P_{\mathrm{min}}=0$, $P_{\mathrm{max}}=10$. \tc{black}{For example, for the left matrix, when the power profile is $(p_1, p_2) = (P_{\max}, P_{\max})$ the individual SINRs are given by: $(\mathrm{SINR}_1, \mathrm{SINR}_2) = \left( \frac{ P_{\max} g_{\min}}{1 + P_{\max} g_{\max}},  \frac{P_{\max} g_{\max}}{1 + P_{\max} g_{\min}} \right)$.}}
\label{fig-PCGame} 
\end{center}
\end{figure}

The coordination code of length $3$ that we develop next can be seen as a separate source channel code, which consists of a source code with distortion and a channel code with side information. The source encoder and decoder are defined by the mappings
\begin{equation}
\begin{array}{cccc}
  f_{\mathrm{S}}: & \calX_0^3 & \rightarrow & \{m_0,m_1\} \\
  & \ul{x}_0 & \mapsto  & i \\
\end{array}, 
\end{equation}
\begin{equation}
\begin{array}{cccc}
  g_{\mathrm{S}} : &  \{m_0,m_1\} &\rightarrow \calX_2^3 \\
  & i &\mapsto & \ul{x}_2
 \end{array}.
\end{equation}
\tc{black}{    Note that the chosen source code only uses $2$ messages $\{m_0,m_1\}$ to represent the $8$ possible sequences $\ul{x}_0$. One of the benefits of this restriction is that it becomes computationally feasible to find the best two-message code by brute force enumeration. Finding a more systematic low-compelxity design approach to coordinatinon codes goes beyond the scope of the present work.} The exact choice of $f_{\mathrm{S}}$ and $g_{\mathrm{S}}$ is provided after we describe the channel code. 

In each block $b$, Transmitter 1's channel encoder implements the mapping
\begin{equation}
\ul{x}_1^{(b)} = f_{\mathrm{C}}(\ul{x}_0^{(b)}, \ul{x}_2^{(b)}, i_{b+1})
\end{equation}
where $i_{b+1}=f_{\mathrm{S}}(\ul{x}_0^{(b+1)})$ is the index associated with the sequence $\ul{x}_2^{(b+1)}$. The idea behind the design of the channel encoder $f_{\mathrm{C}}$ is the following. If Transmitter $1$ did not have to transmit the index $i_{b+1}$, its optimal encoding would be to exploit its knowledge of $(x_0^3(b),x_2^3(b))$ to choose the sequence $x_1^3(b)$ resulting in the highest average payoff in block $b$. However, to communicate the index $i_{b+1}$, Transmitter 1 will instead choose to transmit the sequence $x_1^3(b)$ with the highest average payoff in block $b$ if $i_{b+1}=m_0$, or the sequence $\ul{x}_1^{(b)}$ with the second highest average payoff in block $b$ if $i_{b+1}=m_1$. Note that Transmitter 2 is able to perfectly decode this encoding given its knowledge of $\ul{x}_0^{(b)}$, $\ul{x}_2^{(b)}$, and $\ul{x}_1^{(b)}$ at the end of block $b$. Formally, $f_{\mathrm{C}}$ is defined as follows. 
The sequence $\ul{x}_1 $ is chosen as
 \begin{equation}
 \ul{x}_1 = \ul{x}'_1 \oplus \ul{d}
 \end{equation}
where the modulo-two addition is performed component-wise,
 \begin{align}
 \ul{x}'_1 & \in \arg \max_{\ul{x}_1 \in \mc{X}_1^3}  \sum_{n=1}^3  w^{\mathrm{SINR}}(x_{0,n}, {x}_{1,n}, x_{2,n}),\\
 \ul{d} &=(0,0,0)\text{ if $i_{b+1}=m_0$},\\
 d^3& \in \mathop{\arg\max}_{\ul{d} \text{ s.t. }\Omega(\ul{d})=1}  \sum_{n=1}^3  w^{\mathrm{SINR}}(x_{0,n}, {x}'_{1,n}\oplus d_{n}, x_{2,n})\; \text{ if $i_{b+1}=m_1$}
 \end{align}
where $\Omega$ is the Hamming weight function that is, the number of ones in the sequence $\ul{d} \in \{0,1\}^3$. If the argmax set is not a singleton set, we choose the sequence with the smallest Hamming
weight. 

To complete the construction, we must specify how the source code is designed. Here, we choose the mappings $f_{\mathrm{S}}$ and $g_{\mathrm{S}}$ that maximize the expected payoff $\mathbb{E}(w^{\mathrm{SINR}})$ knowing the operation of the channel code. The source code resulting from an exhaustive search is given in Table \ref{Tab:Source code}, and the corresponding channel code is given in Table~\ref{Tab:Channel code}. The detailed expression of the expected payoff required for the search is provided in Appendix~\ref{app:expected-payoff}.

\begin{table}[h]
\caption{Proposed source coding and decoding for $p=\frac{1}{2}$.}\label{Tab:Source code}
\begin{center}
  \begin{tabular}[]{|c|cccccccc|}
    \hline
    $x_0^3$& 000& 001& 010& 011& 100& 101& 110&111\\\hline
    Index $i=f_{\mathrm{S}}(x_0^3)$ & $m_0$ & $m_0$ & $m_0$& $m_0$& $m_0$& $m_0$& $m_1$& $m_1$\\\hline
    $g_{\mathrm{S}}(i)$& 111 & 111 & 111 & 111 & 111 & 111 & 001 & 001\\\hline
  \end{tabular}

\end{center}
\end{table}

\begin{table}[h]
\caption{Proposed channel coding for $p=\frac{1}{2}$.}\label{Tab:Channel code}
\begin{center}
  \begin{tabular}{|c|cc|cc|cc|cc|cc|cc|cc|cc|}
\hline    $x_0^3(b)$ &\multicolumn{2}{c|}{$000$}&\multicolumn{2}{c|}{$001$}&\multicolumn{2}{c|}{$010$}&\multicolumn{2}{c|}{$011$}&\multicolumn{2}{c|}{$100$}&\multicolumn{2}{c|}{$101$}&\multicolumn{2}{c|}{$110$}&\multicolumn{2}{c|}{$111$}\\\hline
    $x_2^3(b)$ &\multicolumn{2}{c|}{$111$}&\multicolumn{2}{c|}{$111$}&\multicolumn{2}{c|}{$111$}&\multicolumn{2}{c|}{$111$}&\multicolumn{2}{c|}{$111$}&\multicolumn{2}{c|}{$111$}&\multicolumn{2}{c|}{$001$}&\multicolumn{2}{c|}{$001$}\\\hline
    $i_{b+1}$    &$m_0$&$m_1$&$m_0$&$m_1$&$m_0$&$m_1$&$m_0$&$m_1$&$m_0$&$m_1$&$m_0$&$m_1$&$m_0$&$m_1$&$m_0$&$m_1$\\\hline
    $x_1^3(b)$&$000$&$001$&$001$&$000$&$010$&$000$&$011$&$001$&$100$&$000$&$101$&$001$&$110$&$111$&$111$&$110$\\\hline
  \end{tabular}


\end{center}
\end{table}
The proposed codes admit to an intuitive interpretation. For instance, the first line of Table \ref{Tab:Channel code} indicates that if the channel is bad for Transmitter 1 for the three stages of block $b$, then Transmitter $1$ remains silent over the three stages of the block while  Transmitter $2$ transmits at all three stages. In contrast, the last line of Table \ref{Tab:Channel code} shows that if the channel is good for Transmitter 1 for the three stages of block $b$, then Transmitter $1$ transmit at all stages while Transmitter $2$ remains silent two thirds of the time. While this is suboptimal for this specific global wireless channel state realization, this is required to allow coordination and average optimality of the code. 

To conclude this section, we compare the performance of this short code with the best possible performance that would be obtained with infinitely long codes. As illustrated in Fig.~\ref{fig111}, while the performance of the short code suffers from a small penalty compared to that of ideal codes with infinite block length, it still offers a significant gain \ac{wrt} the \ac{SPC} policy and it outperforms the \ac{NPC} policy.

\begin{figure}[htbp]
  \centering
  \includegraphics[width=0.70\textwidth]{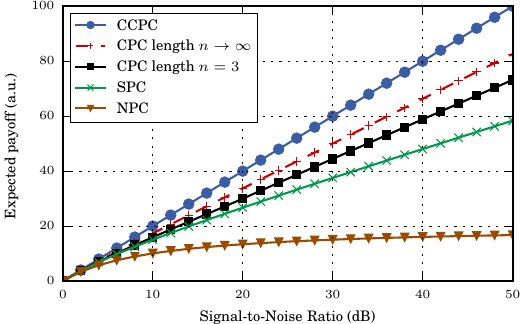}
\caption{Expected payoff versus \ac{SNR} for different power control policies.}
\label{fig111}
\end{figure}

\section{Conclusion}
\label{sec:conclusion}

In this paper, we adopted the view that distributed control policies or resource allocation policies in a network are joint source-channel codes. Essentially, an agent of a distributed network may convey its knowledge of the network state by \emph{encoding} it into a sequence of actions, which can then be \emph{decoded} by the agents observing that sequence. As explicitly shown in Section \ref{sec:coord-scheme}, the purpose of such ``coordination codes'' is neither to convey information reliably nor to meet a requirement in terms of maximal distortion level, but to provide a high expected payoff. Consequently, coordination codes must implement a trade-off between sending information about the future realizations of the network state, which plays the role of an information source and is required to coordinate future actions, and achieving an acceptable payoff for the current state of the network. Considering the large variety of payoff functions in control and resource allocation problems, an interesting issue is whether universal codes performing well within classes of payoff functions can be designed.

Remarkably, since a distributed control policy or resource allocation policy is interpreted as a code, Shannon theory  naturally appears to measure the efficiency of such policies. While the focus of this paper was limited to a small network of two agents, the proposed methodology to derive the best coordination performance in a distributed network is much more general. The assumptions made in this paper are likely to be unsuited to some application scenarios, but provide encouraging preliminary results to further research in this direction. For example, as mentioned in Section \ref{sec:introduction}, a detailed comparison between coded power control and iterative water-filling like algorithms would lead to consider a symmetric observation structure while only an asymmetric structure is studied in this paper. The methodology to assess the performance of good coded policies consists in deriving the right information constraint(s) by building the proof on Shannon theory for the problem of multi-source coding with distortion over multi-user channels wide side information and then to use this constraint to find an information-constrained maximum of the payoff (common payoff case) or the set of Nash equilibrium points which are compatible with the constraint (non-cooperative game case). As a key observation of this paper, the observation structure of a multi-person decision-making problem corresponds in fact to a multiuser channel. Therefore, multi-terminal Shannon theory is not only relevant for pure communication problems but also for any multi-person decision-making problem. The above observation also opens new challenges for Shannon-theorists since decision-making problems define new communication scenarios.

\appendices

\section{Achievable empirical coordination and implementability}
\label{sec:proof-prop-empirical-implementable}

Assume $\ol{Q}$ is an achievable empirical coordination. Then,
for any $\epsilon>0$,
    \begin{align}
      ||\mathbb{E}(\mathrm{T}_{X^N})-\ol{Q}||_1
      &\leq \mathbb{E}(||\mathrm{T}_{X^N}-\ol{Q}||_1)\\
      &= \mathbb{E}(||\mathrm{T}_{X^N}-\ol{Q}||_1|\;
      ||\mathrm{T}_{X^N}-\ol{Q}||_1\geq\epsilon)\mathrm{P}(||\mathrm{T}_{X^N}-\ol{Q}||_1
      \geq\epsilon)\nonumber \\
      &\phantom{=}+\mathbb{E}(||\mathrm{T}_{X^N}-\ol{Q}||_1|\; ||\mathrm{T}_{X^N}-\ol{Q}||_1<\epsilon)\mathrm{P}(||\mathrm{T}_{X^N}-
      \ol{Q}||_1<\epsilon)\\
      &\leq 2\mathrm{P}(||\mathrm{T}_{X^N}-\ol{Q}||_1\geq\epsilon)
      + \epsilon.
    \end{align}
Hence, $\forall \epsilon>0\quad \lim_{N\rightarrow\infty}||\mathbb{E}(\mathrm{T}_{X^N})-\ol{Q}||_1\leq
 \epsilon$, which means that $\ol{Q}$ is implementable.

\section{Lemmas used in proof of Theorem~\ref{thm:achiev-GP}}
\label{app:theo-GP}

\subsection{Proof of Lemma~\ref{lmc:bound_norm_1}}
Recall that $N=\alpha m+(B-1)m$ with our coding scheme. Note that
\begin{align}
  &  \Vert \mathrm{T}_{\ul{x}_0^N\ul{x}_1^N\ul{x}_2^N}-\overline{Q}\Vert_1\nonumber\\
  &\phantom{=}=\sum_{x_0,x_1,x_2}\left\vert\sum_{n=1}^N\frac{1}{N}\mathbbm{1}_{\left\{(x_{0,n},x_{1,n},x_{2,n})=(x_0,x_1,x_2)\right\}}-\ol{Q}(x_0,x_1,x_2)\right\vert\\
  &\phantom{=}= \sum_{x_0,x_1,x_2}\left\vert\sum_{n=1}^{\alpha m}\frac{1}{N}\mathbbm{1}_{\left\{(x^{(1)}_{0,n},x^{(1)}_{1,n},x^{(1)}_{2,n})=(x_0,x_1,x_2)\right\}} +\sum_{b=2}^{B}\sum_{n=1}^{m}\frac{1}{N}\mathbbm{1}_{\left\{(x^{(b)}_{0,n},x^{(b)}_{1,n},x^{(b)}_{2,n})=(x_0,x_1,x_2)\right\}}-\ol{Q}(x_0,x_1,x_2)\right\vert\\
&\phantom{=}\leq \sum_{x_0,x_1,x_2}\left\vert\sum_{n=1}^{\alpha m}\frac{1}{N}\mathbbm{1}_{\left\{(x^{(1)}_{0,n},x^{(1)}_{1,n},x^{(1)}_{2,n})=(x_0,x_1,x_2)\right\}} -\frac{\alpha m}{N}\ol{Q}(x_0,x_1,x_2)\right\vert \nonumber \\
  &\phantom{=}\phantom{======}+ \sum_{b=2}^{B}\sum_{x_0,x_1,x_2}\left\vert\sum_{n=1}^{m}\frac{1}{N}\mathbbm{1}_{\left\{(x^{(b)}_{0,n},x^{(b)}_{1,n},x^{(b)}_{2,n})=(x_0,x_1,x_2)\right\}}-\frac{N-\alpha m}{N(B-1)}\ol{Q}(x_0,x_1,x_2)\right\vert \label{eq:subadd}\\
  &\phantom{=}= \frac{\alpha m}{N}\sum_{x_0,x_1,x_2}\left\vert\sum_{n=1}^{\alpha m}\frac{1}{\alpha m}\mathbbm{1}_{\left\{(x^{(1)}_{0,n},x^{(1)}_{1,n},x^{(1)}_{2,n})=(x_0,x_1,x_2)\right\}} -\ol{Q}(x_0,x_1,x_2)\right\vert \nonumber \\
  &\phantom{=}\phantom{======}+ \frac{m}{N}\sum_{b=2}^{B}\sum_{x_0,x_1,x_2}\left\vert\sum_{n=1}^{m}\frac{1}{m}\mathbbm{1}_{\left\{(x^{(b)}_{0,n},x^{(b)}_{1,n},x^{(b)}_{2,n})=(x_0,x_1,x_2)\right\}}-\ol{Q}(x_0,x_1,x_2)\right\vert \\
  &\phantom{=}\leq \frac{2\alpha}{B-1+\alpha}+\frac{1}{B-1} \sum_{b=2}^{B} \Vert \mathrm{T}_{\ul{x}_0^{(b)}\ul{x}_1^{(b)}\ul{x}_2^{(b)}}-\overline{Q}\Vert_1, \label{eq:lasteq-bound}
\end{align}
 where \eqref{eq:subadd} follows from the triangle inequality and \eqref{eq:lasteq-bound} follows from $\frac{m}{N}\leq \frac{1}{B-1}$ and $\Vert P-Q\Vert_1\leq 2$ for $(P,Q)\in\Delta^2(\calX)$.

\subsection{Proof of Lemma~\ref{lm:lemma_case_2_event_0}}
The proof of this result is similar to that of the following Lemmas, using the ``uncoordinated'' distribution $\widehat{Q}$ instead of $Q$. For brevity, we omit the proof.

\subsection{Proof of Lemma~\ref{lm:lemma_case_2_event_1}}
Note that
\begin{align}
  \mathbb{E}\left(\mathrm{P}(E_1^{(b)})\right) = \mathrm{P}\left((\ul{X}_0^{(b+1)},\ul{X}_2(i'_b))\notin\calT^n_{\epsilon_1}(Q_{X_0X_2})\text{ for all } i'_b\in[1:2^{nR}]\right)
\end{align}
with $\ul{X}_0^{(b)}$ distributed according to $\prod_{n=1}^n Q_{X_0}$, and $\ul{X}_2(i_b)$ independent of each other distributed according to $\prod_{n=1}^m Q_{X_2}$. Hence, the result directly follows from the covering lemma~\cite[Lemma 3.3]{NetworkInformationTheory}, with the following choice of parameters.
\begin{align*}
  U&\leftarrow\emptyset\qquad X^n\leftarrow \ul{X}_0^{(b+1)}\qquad \hat{X}^n(m)\text{ with }m\in\mathcal{A}\leftarrow \ul{X}_2(i'_b)\text{ with }i'_b\in[1:2^{nR}].
\end{align*}
\subsection{Proof of Lemma~\ref{lm:lemma_case_2_event_2}}
Note that
\begin{align}
  \mathbb{E}\left(\mathrm{P}(E_2^{(b)}|E_1^{(b-1)c})\right)
=\mathrm{P}\left((\ul{U}(I_b,j'_b),\ul{X}^{(b)}_0,\ul{X}_2(I_{b-1}))\notin\calT^n_{\epsilon_2}(Q_{UX_0X_2}) \text{ for all } j'_b\in[1:2^{nR'}]\right),
\end{align}
where $(\ul{X}_0^{(b)},\ul{X}_2(I_{b-1}))\in\calT^n_{\epsilon_1}(Q_{X_0X_2})$, and $\ul{U}(I_b,j)$ are generated independently of each other according to $\prod_{n=1}^n Q_U$. Hence, the result follows directly from the covering lemma~\cite[Lemma 3.3]{NetworkInformationTheory} with the following choice of parameters.
\begin{align*}
  U\leftarrow \emptyset\qquad X^n\leftarrow \ul{X}_0^{(b)},\ul{X}_2(I_{b-1})\qquad \hat{X}^n(m)\text{ with $m\in\mathcal{A}$}\leftarrow \ul{U}(I_b,j'_b)\text{ with $j'_b\in[1:2^{nR'}]$}.
\end{align*}

\subsection{Proof of Lemma~\ref{lm:lemma_case_2_event_3}}

The result follows from a careful application of the conditional typicality lemma. Note that conditioning on $E_{2}^{(b)c}$ ensures that $(\ul{u}(I_b,J_b),\ul{X}_0^{(b)},\ul{x}_2(I_{b-1}))\in\calT_{\epsilon_2}^n(Q_{UX_0X_2})$, while conditioning on $E_2^{(b-1)c}\cap E_3^{(b-1)c}\cap E_4^{(b-1)c}\cap E_0^{c}$ guarantees that $\widehat{I}_{b-1}={I}_{b-1}$. Consequently,
\begin{align}
&\mathrm{P}\left(E_3^{(b)}\cap E_2^{(b)c}\cap E_2^{(b-1)c}\cap E_3^{(b-1)c}\cap E_4^{(b-1)c}\cap E_0^c\right) \nonumber \\
&\phantom{=}\leq \mathrm{P}\left( (\ul{u}(I_b,J_b),\ul{X}_0^{(b)},\ul{X}_1^{(b)},\ul{x}_2(\widehat{I}_{b-1}),\ul{Y}^{(b)} ) \notin\calT_{\epsilon_3}^n(Q_{X_0X_2X_1Y})\right. \nonumber\\
&\phantom{===================}\left.\vert (\ul{u}(I_b,J_b),\ul{X}_0^{(b)},\ul{x}_2(I_{b-1}))\in\calT_{\epsilon_2}^n(Q_{UX_0X_2})\cap \widehat{I}_{b-1}=I_{b-1}\right) \\
&\phantom{=}=\sum_{i_{b-1},i_{b},j_b,\ul{x}_0^{(b)}}p_{\widehat{I}_{b-1},I_b,J_b,\ul{X}_0^{(b)}}(i_{b-1},i_b,j_b,\ul{x}_0)\sum_{\ul{y}}\sum_{\ul{x}_1^{(b)}}\Gamma(\ul{y}|\ul{x}_0^{(b)},\ul{x}^{(b)}_1,\ul{x}_2(i_{b-1}))\nonumber\\
&\phantom{==========}  Q(\ul{x}^{(b)}_1|\ul{u}(i_b,j_b),\ul{x}_0^{(b)}, \ul{x}_2(i_{b-1}))\mathbbm{1}_{\left\{(\ul{u}(i_b,j_b),\ul{x}_0^{(b)},\ul{x}_1^{(b)},\ul{x}^{(b)}_2(i_{b-1}),\ul{y})\notin\calT_{\epsilon_3}^n(Q_{UX_0X_1X_2Y})\right\}},
\end{align}
where $p_{\widehat{I}_{b-1},I_b,J_b,\ul{X}_0^{(b)}}$ denotes the joint distribution of ${\widehat{I}_{b-1},I_b,J_b,\ul{X}_0^{(b)}}$ given $(\ul{u}(I_b,J_b),\ul{X}_0^{(b)},\ul{x}_2(I_{b-1}))\in\calT_{\epsilon_2}^n(Q_{UX_0X_2})$ and $\widehat{I}_{b-1}=I_{b-1}$. Upon taking the average over the random codebooks, we obtain
\begin{multline}
  \mathbb{E}\left(\mathrm{P}\left(E_3^{(b)}|E_1^{(b-1)c}\cap E_2^{(b-1)c}\cap E_3^{(b-1)c}\cap E_4^{(b-1)c}\cap E_0^c\right)\right)\\
  =\sum_{i_{b-1},i_{b},j_b,\ul{x}_0^{(b)}}\mathbb{E}\left(p_{\widehat{I}_{b-1},I_b,J_b,\ul{X}_0^{(b)}}(i_{b-1},i_b,j_b,\ul{x}_0^{(b)}) \mathbb{E}\left(\sum_{\ul{y}}\sum_{\ul{x}_1^{(b)}}\Gamma(\ul{y}|\ul{x}_0^{(b)},\ul{x}^{(b)}_1,\ul{X}_2(i_{b-1}))\right.\right.\\
    \left.\left.  Q(\ul{x}^{(b)}_1|\ul{U}(i_b,j_b),\ul{x}_0^{(b)}, \ul{X}_2(i_{b-1}))\mathbbm{1}_{\left\{(\ul{u}(i_b,j_b),\ul{x}_0^{(b)},\ul{x}_1^{(b)},\ul{X}^{(b)}_2(i_{b-1}),\ul{y})\notin\calT_{\epsilon_3}^n(Q_{UX_0X_1X_2Y}\right\}}\big\vert \ul{U}(i_b,j_b),\ul{X}_2(i_{b-1})) \right)\right).
\end{multline}
The inner expectation is therefore
\begin{align}
  \mathrm{P}\left((\ul{u},\ul{x}_0,\ul{X}_1,\ul{x}_2,\ul{Y})\notin\calT_{\epsilon_3}^n(Q_{UX_0X_1X_2Y})\right),\label{eq:inner_expect}
\end{align}
where $(\ul{X}_1,\ul{Y})$ is distributed according to $\prod_{n=1}^m Q_{X_1|UX_0X_2}(x_{1,n}|u_{1,n},x_{2,n}x_{0,n})\Gamma(y_n|x_{0,n} x_{1,n}x_{2,n})$ given $(\ul{u},\ul{x}_0,\ul{x}_2)\in\calT_{\epsilon_2}^n(Q_{X_0X_2})$. The conditional typicality lemma~\cite[p. 27]{NetworkInformationTheory} guarantees that~\eqref{eq:inner_expect} vanishes as $n\rightarrow\infty$.

\subsection{Proof of Lemma~\ref{lm:lemma_case_2_event_4}}

The result is a consequence of the packing lemma. Note that
\begin{align}
&  \mathbb{E}\left(\mathrm{P}\left(E_4^{(b)}\cap E_2^{(b-1)c} \cap E_3^{(b-1)c} \cap E_4^{(b-1)c} \cap E_0^c\right)\right)\nonumber\\
&\phantom{===}\leq \mathbb{E}\left(\mathrm{P}\left(E_4^{(b)}|\widehat{I}_{b-1}={I}_{b-1}\right)\right)\\
  &\phantom{===}=\mathrm{P}\left((\ul{U}(i'_b,j'_b),\ul{X}_2(\widehat{I}_{b-1}),\ul{Y}^{(b)})\in \calT^n_{\epsilon_2}(Q)\text{ for some } (i'_b,j'_b)\text{ with  }i'_b \neq I_b|\widehat{I}_{b-1}={I}_{b-1}\right)
\end{align}
since conditioning on $E_2^{(b-1)c}\cap E_3^{(b-1)c}\cap E_4^{(b-1)c}\cap E_0^{c}$ guarantees that $\widehat{I}_{b-1}={I}_{b-1}$. Since every $\ul{U}(i'_b,j'_b)$ with $i'_b\neq I_b$ is generated according to $\prod_{i=1}^mp_U(u_i)$ independently of $(\ul{Y},\ul{X}_2(I_{b-1}))$, byt the packing lemma~\cite[Lemma 3.1]{NetworkInformationTheory} we know that if $R+R'<I_Q(U;Y,X_2)-\delta(\epsilon_3)$ then
\begin{align*}
\mathrm{P}\left((\ul{U}(i'_b,j'_b),\ul{X}_2(\widehat{I}_{b-1}),\ul{Y}^{(b)})\in \calT^n_{\epsilon_2}(Q)\text{ for some } (i'_b,j'_b)\text{ with  }i'_b \neq I_b|\widehat{I}_{b-1}={I}_{b-1}\right)
\end{align*}
vanishes as $n\rightarrow\infty$.

\section{Proof of Lemma \ref{lemma:convexity-of-phi}}
\label{sec:proof-lemma-refl}

The function $\Phi = \Phi^{\mathrm{I}}$ can be rewritten as $\Phi(Q) = H_Q(X_0) - H_Q(Y,X_0 | X_2) + H_Q(Y|X_0,X_2,X_1)$. The first term $H_Q(X_0) = -\sum_{x_0}
\rho_0(x_0) \log \rho_0(x_0)$ is a constant w.r.t. $Q$. The third term is linear w.r.t. $Q$ since,  with $\Gamma$ fixed,
\begin{align}H_Q(Y|X_0,X_2,X_1) &= - \sum_{x_0,x_1,x_2,y} Q(x_0,x_1,x_2,y) \log \Gamma(y|x_0,x_1,x_2).
\end{align}
It is therefore sufficient to prove that $H_Q(Y,X_0 | X_2)$ is
concave. Let $\lambda_1 \in [0,1]$, $\lambda_2 = 1 -
\lambda_1$, $(Q_1,Q_2) \in  \Delta^2(\mc{X}_0 \times
\mc{X}_1 \times \mc{X}_2 \times \mc{Y})$ and
$Q=\lambda_1 Q_1 + \lambda_2 Q_2$. We have that:
\begin{align}
H_Q(Y,X_0 | X_2) &=  - \sum_{x_0,x_2,y} \bigg( \sum_{x_1,i} \lambda_i Q_i(x_0,x_1,x_2,y) \bigg)\log \left[ \frac{\sum_{x_1,i} \lambda_i Q_i(x_0,x_1,x_2,y)}{ \sum_{i} \lambda_i Q_i(x_2)} \right] \\\displaybreak[0]
  &= - \sum_{x_0,x_2,y} \bigg(\sum_i \lambda_i  \sum_{x_1} Q_i(x_0,x_1,x_2,y) \bigg) \log \left[ \frac{\sum_i \lambda_i  \sum_{x_1} Q_i(x_0,x_1,x_2,y)}{ \sum_{i} \lambda_i Q_i(x_2)} \right] \\
  &> - \sum_i \lambda_i \sum_{x_0,x_2,y} \bigg(  \sum_{x_1} Q_i(x_0,x_1,x_2,y) \bigg) \log \left[ \frac{\lambda_i  \sum_{x_1} Q_i(x_0,x_1,x_2,y)}{\lambda_i Q_i(x_2)} \right] \\
&= - \sum_i \lambda_i \sum_{x_0,x_2,y} \bigg(  \sum_{x_1} Q_i(x_0,x_1,x_2,y) \bigg) \log \left[ \frac{\sum_{x_1} Q_i(x_0,x_1,x_2,y)}{ Q_i(x_2)} \right] \\
&= \lambda_1 H_{Q_1}(Y,X_0 | X_2) + \lambda_2 H_{Q_2}(Y,X_0 | X_2)
\end{align}
where the strict inequality comes from the log-sum inequality~\cite{Cover:2006:EIT:1146355}, with:
\begin{equation} a_i = \lambda_i Q_i(x_0,x_1,x_2)
\end{equation}
and
\begin{equation} b_i = \lambda_i Q_i(x_2)
\end{equation}
for $i\in\{1,2\}$ and for all $(x_0,x_1,x_2)$
 such that $Q_i(x_2)>0$.

\section{Expression of the expected payoff $W_3$ which allows
the best mappings $f_{\mathrm{S}}$ and $g_{\mathrm{S}}$ to be selected}
\label{app:expected-payoff}

We introduce the composite mapping $\chi_{\mathrm{S}} = g_{\mathrm{S}} \circ f_{\mathrm{S}}$. For the channel code defined in Section \ref{sec:coord-scheme}, the expected payoff only depends on the mappings $f_{\mathrm{S}}$ and $\chi_{\mathrm{S}}$, and we denote it by $W_3(f_{\mathrm{S}}, \chi_{\mathrm{S}})$. The following notation is used below: $\chi_{\mathrm{S}}(x_0^3) = (\chi_{1}(x_0^3), \chi_{2}(x_0^3), \chi_{3}(x_0^3))$ to stand for the three components of $\chi_{\mathrm{S}}$.

It can be checked that 
\begin{equation}
W_3(f_{\mathrm{S}}, \chi_{\mathrm{S}}) = \sum_{(i,j,k) \in
\{0,1\}^3 } W_{ijk}(f_{\mathrm{S}}, \chi_{\mathrm{S}})
\end{equation}
where:
\begin{align} W_{000}(f_{\mathrm{S}}, \chi_{\mathrm{S}}) &=
\underbrace{p^3}_{\mathrm{P}[x_0^3[b]= (0,0,0)]}
\bigg[ \underbrace{P_0(f_{\mathrm{S}}) }_{\mathrm{P}[i_{b+1} = m_0]} \times\underbrace{\left(\max_{x_1^3 \in \{0,1\}^3} \{ \sum_{n=1}^3 w(0, x_{1,n}, \chi_n((0,0,0))) \}\right)}_{\textrm{Best payoff}}  \nonumber \\
& + (1-P_0(f_{\mathrm{S}}) ) \times \underbrace{\left(\max_{x_1^3 \in \{0,1\}^3\backslash \mc{X}_{1}^{000}} \{ \sum_{n=1}^3 w(0, x_{1,n},  \chi_n((0,0,0))) \}\right)}_{\textrm{Second best payoff}}\bigg],
  \end{align}
\begin{align} \mc{X}_{1}^{000} = \arg\max_{x_1^3 \in \{0,1\}^3} \{
 \sum_{n=1}^3 w(0, x_{1,n}, \chi_n((0,0,0))) \},\end{align}
\begin{align} W_{001}(f_{\mathrm{S}}, \chi_{\mathrm{S}})&= p^2(1-p)  \bigg[ P_0(f_{\mathrm{S}}) \times \Bigg(\max_{x_1^3 \in \{0,1\}^3} \Big\{   w(0,x_{1,1}, \chi_1((0,0,1))) + w(0,x_{1,2}, \chi_2((0,0,1))) \nonumber \\
& \phantom{===========} + w(1, x_{1,3}, \chi_3((0,0,1))) \Big\}\Bigg) \nonumber \\
& +  (1-P_0(f_{\mathrm{S}})) \times \Big(\max_{x_1^3 \in \{0,1\}^3\backslash \mc{X}_{1}^{001} } \Big\{ w(0,x_{1,1},\chi_1((0,0,1))) + w(0,x_{1,2}, \chi_2((0,0,1))) \nonumber \\
      & \phantom{==========} + w(1, x_{1,3}, \chi_3((0,0,1)) ) \Big\}\Big) \bigg],
      \end{align}
\begin{align} & \mc{X}_{1}^{001}= \arg\max_{x_1^3 \in \{0,1\}^3} \{ w(0,x_{1,1},\chi_1((0,0,1)) ) + w(0,x_{1,2},  \chi_2((0,0,1))) + w(1,x_{1,3}, \chi_3((0,0,1))  ) \},\end{align}
      \begin{align} W_{010}(f_{\mathrm{S}},\chi_{\mathrm{S}}) &= p^2 (1-p)  \bigg[ P_0(f_{\mathrm{S}}) \times  \Bigg(\max_{x_1^3 \in \{0,1\}^3}  \Big\{ w(0,x_{1,1}, \chi_1(0,1,0)) + w(1, x_{1,2}, \chi_2(0,1,0)) \nonumber \\
 & \phantom{===========} + w(0, x_{1,3},  \chi_3(0,1,0)  ) \Big\}\Bigg) \nonumber \\
      & +  (1-P_0(f_{\mathrm{S}})) \times \Big(\max_{x_1^3 \in \{0,1\}^3\backslash
       \mc{X}_1^{010}}\Big\{ w(0, x_{1,1}, \chi_1(0,1,0) ) + w(1, x_{1,2},
        \chi_2(0,1,0) ) + \nonumber \\
      &\phantom{========} w(0, x_{1,3}, \chi_3(0,1,0)  ) \Big\}\Big) \bigg],
      \end{align}
       \begin{align} \mc{X}_1^{010} = \arg\max_{x_1^3 \in \{0,1\}^3} \{ w(0, x_{1,1}, \chi_1(0,1,0)) + w(1, x_{1,2},   \chi_2(0,1,0) ) + w(0, x_{1,3}, \chi_3(0,1,0) )\},\end{align}
 \begin{align} W_{100}(f_{\mathrm{S}},
      \chi_{\mathrm{S}}) & = p^2(1-p)  \bigg[ P_0(f_{\mathrm{S}}) \times \Bigg(\max_{x_1^3 \in \{0,1\}^3} \Big \{  w(1, x_{1,1}, \chi_1((1,0,0))  ) +   w(0, x_{1,2},  \chi_2((1,0,0)) ) \nonumber \\
       & \phantom{========} +    w(0, x_{1,3}, \chi_3((1,0,0))  ) \}\Bigg) \nonumber \\
      & +  (1-P_0(f_{\mathrm{S}})) \times
       \Big(\max_{x_1^3 \in \{0,1\}^3\backslash \mc{X}_1^{100}}
        \{ w(1, x_{1,1}, \chi_1((1,0,0)) ) + w(0, x_{1,2}, \chi_2((1,0,0)) ) \nonumber \\
      &\phantom{========} +  w(0, x_{1,3}, \chi_3((1,0,0))   ) \}\Big) \bigg],
      \end{align}
\begin{align} \mc{X}_1^{100} = \arg\max_{x_1^3 \in \{0,1\}^3} \{
w(1, x_{1,1}, \chi_1((1,0,0))) + w(0, x_{1,2}, \chi_2((1,0,0))  ) +  w(0, x_{1,3}, \chi_3((1,0,0))) \},\end{align}
 \begin{align} W_{111}(f_{\mathrm{S}},
      \chi_{\mathrm{S}}) &= (1-p)^3 \bigg[ P_0(f_{\mathrm{S}}) \times \Bigg(\max_{x_1^3 \in \{0,1\}^3} \Big\{ \sum_{n=1}^3   w(1, x_{1,n}, \chi_n((1,1,1)) ) \big\}\Bigg) \nonumber \\
  & + (1-P_0(f_{\mathrm{S}}))
  \times \Bigg(\max_{x_1^3 \in
   \{0,1\}^3\backslash \mc{X}_1^{111}}\Big\{ \sum_{n=1}^3 w(1, x_{1,n}, \chi_n((1,1,1)))
    \Big\}\Bigg) \bigg],
  \end{align}
  \begin{align} \mc{X}_1^{111} =
\arg\max_{x_1^3 \in \{0,1\}^3} &\{ \sum_{n=1}^3
 w(1, x_{1,n}, \chi_n((1,1,1)) ) \},
  \end{align}
   \begin{align} W_{011}(f_{\mathrm{S}} \chi_{\mathrm{S}}) &= p(1-p)^2  \Bigg[ P_0(f_{\mathrm{S}}) \times \bigg(\max_{x_1^3 \in \{0,1\}^3} \Big\{ w(0, x_{1,1}, \chi_1((0,1,1))) + w(1, x_{1,2}, \chi_2((0,1,1)) ) \nonumber \\
 &\phantom{==========}+ w(1, x_{1,3}, \chi_3((0,1,1))) \Big\}\bigg) \nonumber \\
      & +  (1-P_0((f_{\mathrm{S}})))\times  \bigg(\max_{x_1^3
       \in \{0,1\}^3\backslash \mc{X}_1^{011}}\Big\{ w(0, x_{1,1},
       \chi_1((0,1,1))) + w(1, x_{1,2}, \chi_2((0,1,1)))+ \nonumber \\
      & \phantom{=========} w(1, x_{1,3}, \chi_3((0,1,1))) \Big\}\bigg) \Bigg],
      \end{align}
\begin{align} \mc{X}_1^{011} =\arg\max_{x_1^3 \in \{0,1\}^3} \{ w(0, x_{1,1}, \chi_1((0,1,1))) + w(1, x_{1,2}, \chi_2((0,1,1))) + w(1, x_{1,3}, \chi_3((0,1,1))) \},
\end{align}
  \begin{align}  W_{101}(f_{\mathrm{S}},\chi_{\mathrm{S}}) &= p(1-p)^2  \Bigg[ P_0(f_{\mathrm{S}}) \times \bigg(\max_{x_1^3 \in \{0,1\}^3} \Big\{ w(1, x_{1,1}, \chi_1((1,0,1))) + w(0, x_{1,2},  \chi_2((1,0,1)) ) \nonumber \\
      &\phantom{==========} + w(1, x_{1,3}, \chi_3((1,0,1))  )  \Big\}\bigg) \nonumber \\
      & +  (1-P_0(f_{\mathrm{S}})) \times  \bigg(\max_{x_1^3 \in \{0,1\}^3\backslash
      \mc{X}_1^{101}} \Big\{ w(1, x_{1,1},  \chi_1((1,0,1)) ) + w(0, x_{1,2}, \chi_2((1,0,1))) + \nonumber \\
      &\phantom{===========}  w(1, x_{1,3},  \chi_3((1,0,1))) \Big\}\bigg) \Bigg],
      \end{align}
     \begin{align} \mc{X}_1^{101} =\arg\max_{x_1^3 \in \{0,1\}^3} \Big\{ w(1, x_{1,1}, \chi_1((1,0,1))  ) +w(0, x_{1,2}, \chi_2((1,0,1))) + w(1, x_{1,3}, \chi_3((1,0,1))  ) \Big\},\end{align}
  \begin{align} W_{110}(f_{\mathrm{S}},
      \chi_{\mathrm{S}})&= p(1-p)^2  \Bigg[ P_0(f_{\mathrm{S}}) \times \bigg(\max_{x_1^3 \in \{0,1\}^3}\Big\{ w(1, x_{1,1}, \chi_1((1,1,0))) + w(1, x_{1,2}, \chi_2((1,1,0)))  \nonumber \\
      &\phantom{============}+ w(0, x_{1,3}, \chi_3((1,1,0)))  \Big\}\bigg) \nonumber \\
      & +  (1- P_0(f_{\mathrm{S}})) \times \bigg(\max_{x_1^3 \in \{0,1\}^3
      \backslash \mc{X}_1^{110}} \Big\{ w(1, x_{1,1}, \chi_1((1,1,0)) ) + w(1, x_{1,2}, \chi_2((1,1,0))) \nonumber \\
      &\phantom{============} + w(0, x_{1,3}, \chi_3((1,1,0))) \Big\}\bigg) \Bigg],  \end{align}
\begin{align}  \mc{X}_1^{110}= \arg\max_{x_1^3 \in \{0,1\}^3} \{ w(1, x_{1,1},  \chi_1((1,1,0)) ) + w(1, x_{1,2}, \chi_2((1,1,0))  ) + w(0, x_{1,3}, \chi_3((1,1,0)))\}. \end{align}

In the case of Table \ref{Tab:Source code}, $P_0(f_{\mathrm{S}})$ is given by
 \begin{align} P_0(f_{\mathrm{S}}) & =
     \mathrm{P}\big[ x_0^3=(0,0,0) \big]+ \mathrm{P}\big[ x_0^3=(0,0,1) \big] + \mathrm{P}\big[ x_0^3=(0,1,0) \big]+ \mathrm{P}\big[ x_0^3=(0,1,1) \big] \nonumber \\
    &+\mathrm{P}\big[ x_0^3=(1,0,0) \big] +\mathrm{P}\big[ x_0^3=(1,0,1) \big]\\
    &= p(2- p)\\
    & \stackrel{p=\frac{1}{2}}{=} \frac{3}{4}.
    \end{align}

\bibliographystyle{IEEEtran}
\bibliography{bib-TIT-v6}

\end{document}